\documentclass[11pt]{article}

\usepackage[utf8]{inputenc}
\usepackage[T1]{fontenc}
\usepackage{fullpage}
\usepackage{amsmath, amssymb, amsthm, amsfonts, bm, mathtools, physics}
\usepackage{mleftright}\mleftright
\usepackage[pagebackref,final]{hyperref}\hypersetup{breaklinks=true, colorlinks = true, allcolors = blue}
\usepackage{thmtools}
\usepackage{thm-restate}
\usepackage{enumitem}
\usepackage[textsize=tiny]{todonotes}
\usepackage{algorithm}
\allowdisplaybreaks

\usepackage{authblk}
\usepackage{cleveref}
\makeatletter
\AddToHook{cmd/appendix/before}{\def\cref@section@alias{appendix}}
\makeatother
\usepackage{graphicx}
\usepackage{caption}
\usepackage{float}

\usepackage{silence}
\usepackage{ifdraft}
\ifdraft{
	
}{
	
}

\newtheorem{theorem}{Theorem}

\newtheorem{remark}{Remark}
\newtheorem{lemma}[theorem]{Lemma}
\newtheorem{fact}[theorem]{Fact}
\newtheorem{corollary}[theorem]{Corollary}
\newtheorem{definition}{Definition}
\newtheorem{claim}{Claim}

\newcommand{\eps}{\epsilon}

\newcommand{\bigO}[1]{\mathcal{O}\left( #1 \right)}

\newcommand{\pgood}{\Pi_\text{good}}

\newcommand{\pzeroa}{\Pi_{0^a}}
\newcommand{\pperp}{\Pi_{\perp}}

\newcommand{\circclass}{\textsf{MCM}}

\newcommand{\embe}{\textsf{EMBE}}
\newcommand{\ambe}{\textsf{AMBE}}

\newcommand{\addm}{\textsf{ADD}$_m$}
\newcommand{\addp}{\textsf{ADD}_p}

\newcommand{\err}{\mathcal{E}_{BE}}
\newcommand{\seq}{\mathcal{S}_x}

\newcommand{\herm}{\mathcal{H}}
\newcommand{\general}{\mathcal{A}}

\newcommand{\calD}{\mathcal{D}}

\newcommand{\calH}{\mathcal{H}}

\newcommand{\calM}{\mathcal{M}}

\newcommand{\calO}{\mathcal{O}}

\newcommand{\calS}{\mathcal{S}}
\newcommand{\calT}{\mathcal{T}}

\newcommand{\good}{\Pi_\textnormal{good}}
\newcommand{\bad}{\Pi_\textnormal{bad}}
\newcommand{\spant}{\textnormal{span}}

\renewcommand{\epsilon}{\varepsilon}

\title{Methods for Reducing Ancilla-Overhead in Block Encodings}

\author[1]{Francisca Vasconcelos\footnote{francisca@berkeley.edu}}
\author[2]{András Gilyén\footnote{gilyen.andras@renyi.hu}}
\affil[1]{University of California, Berkeley}
\affil[2]{HUN-REN Alfréd Rényi Institute of Mathematics}
\date{}

\begin{document}
 \maketitle

\begin{abstract}\footnotesize
    Block encodings are a fundamental primitive in quantum algorithms but often require large ancilla overhead. In this work, we introduce novel techniques for reducing this overhead in two distinct ways.
    
    In Part I, we devise a method that approximately uncomputes all but one of the ancillae (freeing up the ancillae for reuse in later parts of a quantum algorithm) for any block encoding, thereby effectively providing a ``space-time tradeoff''. More precisely, for any matrix $A$, such that $\|A\|\leq 1-\delta$, our protocol coherently computes an $\epsilon$-approximate single-ancilla block-encoding of $A$, constructed through $\calO(\frac{1}{\delta}\log(\frac{1}{\epsilon}))$ queries to any block encoding $U$ and $U^\dagger$. We establish a connection between this procedure and the problem of ``phase correction'' in modular quantum signal processing.

    In Part II, we evaluate the minimum number of ancillae required to perform coherent multiplication of $K$ block encodings and, in doing so, introduce a ``space-accuracy tradeoff''. First, we prove that $\lceil\log_2 K\rceil$ ancillae is optimal for \emph{exact} multiplication of block encodings. This lower-bound is saturated by a slight modification to the ``compression gadget'' of Low and Wiebe (2019). Furthermore, we study $\epsilon$-\emph{approximate} multiplication of block encodings in regimes where the block encodings have $\calO(1/K)$-bounded deviation from the identity and propose the $p$-Modular Addition Compression Gadget which has bounded error $\calO(K^{-2^p})$, using just $p$ additional ancillae. We discuss applications of the $p$-MACG to Hamiltonian simulation and quantum differential equation solvers. Finally, we demonstrate how oblivious amplitude amplification can boost the success probability of an $\epsilon$-approximate multiplcation of block encodings circuit to $1-\calO(\epsilon^2)$.
\end{abstract}


\newpage
{\tableofcontents}

\newpage

\section{Introduction} \label{sec:intro}
Block encodings are a fundamental primitive in quantum algorithms \cite{chakraborty2018BlockMatrixPowers, low2016HamSimQubitization, low2019Hamiltonian, gilyen2018QSingValTransf, dalzell2023QuantumAlgSurvey}, efficiently representing non-unitary linear operations on quantum subsystems. In practice, however, their implementation and manipulation can be resource intensive---often requiring large ancilla overhead \cite{clader2022quantum}. For example, the Quantum Singular Value Transformation (QSVT) \cite{gilyen2018QSingValTransf}---a fundamental quantum subroutine---requires repeated sequential access to a block encoding. Each such access is accompanied with a large generalized Toffoli operation entangling each ancillae, causing a large overhead, e.g., in T count, that scales multiplicatively with the total number of calls and ancillas. Another example includes multiplication of block encodings, an important subroutine for Hamiltonian simulation \cite{berry2015simulating, low2019Hamiltonian} and quantum scientific computing \cite{fang2023time}. Whereas one might naïvely expect $K$ ancillae necessary to multiply $K$ block encodings, Low and Wiebe \cite{low2019Hamiltonian} introduced a ``compression gadget'' reducing the overall ancilla overhead to $\log_2 K$.

This paper offers two novel methods for reducing the ancilla (space) resources for block encodings, applicable to both of these important settings. Although, there has been substantial work on time-accuracy tradeoffs for various quantum algorithms, such as Quantum Signal Processing \cite{low2016HamSimQSignProc} and the Quantum Singular Value Transform \cite{gilyen2018QSingValTransf}, to our knowledge, this paper constitutes the first explicit proposal and study of space-time and space-accuracy tradeoffs for manipulation of block encodings.  

The first part of this work introduces a general subroutine to convert any block-encoding to its ``canonical'' version that uses a single ancilla qubit, and (deterministically) resets all other ancillas to the zero state. This can be viewed as a quantum ``space-time tradeoff'', in which ancilla (space) overhead is reduced in exchange for added computational time. This result has both theoretical and practical implications. Theoretically, a novel algorithm is established for coherently uncomputing all but one of a block encoding's ancillae. To do so, we employ oblivious amplitude amplification as a general ancilla-uncomputation procedure for \emph{unitary} block encoded matrices. For example, given a block encoding of any Hermitian matrix $H$, our algorithm maps it (via Hermitian dilation) into its ``canonical'' single-ancilla unitary block encoding $U_H=\tiny\arraycolsep=0.3\arraycolsep\ensuremath{\begin{pmatrix}H & \sqrt{I-H}\\\sqrt{I-H} & -H\end{pmatrix}}$, which enables oblivious amplitude amplification to return all other ancilla to $\ket{0}$. In practice, the creation and subsequent composition of block encodings can substantially increase ancilla-overhead in quantum algorithms. With our new method accumulation of ancilla overhead can be prevented, while incurring a moderate time overhead. 

The second part of this work introduces a quantum ``space-accuracy tradeoff'', in which ancilla (space) overhead is reduced in exchange for degraded computational accuracy. Notably, we demonstrate that, by relaxing from exact to approximate computation (for structured instances), ancilla-overhead can be substantially reduced. This mirrors the common classical phenomena in which exact resource requirements can be much larger than approximate ones--e.g., the exact versus approximate polynomial degree-gap of Boolean functions\footnote{E.g., exactly computing the OR function on $n$ bits requires degree $n$, while an $\varepsilon$-approximation can be achieved with degree $\Theta(\sqrt{n})$~\cite{nisan1994degree,paturi1992degree}, which is deeply related to the complexity of Gover's search algorithm~\cite{grover1996QSearch}.} or space-gap of streaming algorithms\footnote{E.g., exactly computing frequency moments $F_k$ of a data stream requires $\Omega(n)$ space, 
while certain moments (e.g., $F_0$ and $F_2$) admit $\varepsilon$-approximations in $\mathrm{polylog}(n)$ space~\cite{alon1999space}.}. Although, we expect this to be a broad quantum phenomenon, occurring for many forms of manipulation of block encodings, this work specifically focuses on multiplication of block encodings, as a proof of concept. 

In the multiplication setting, we begin by establishing a $\log_2 K$ ancilla lower-bound for \emph{exact} multiplication of $K$ block encodings, via a Hilbert space dimension-counting argument. This result both establishes optimality of the Low-Wiebe ``compression gadget'' and, more generally, illustrates how information-theoretic  arguments can be used to certify space lower-bounds for manipulation of block encodings.  

We then formalize an appropriate notion of \emph{approximate} multiplication of block encodings so as to achieve an ancilla-overhead reduction substantially surpassing the exact multiplication lower-bound. Specifically, we show that in the large-$K$ regime, for block encodings with bounded distance from the identity matrix, high-precision approximate multiplication of $K$ block encodings can be achieved with just a single ancilla. Conceptually, the proof approach is a product-compression routine that behaves like a quantum streaming ``sketch'' for operator products. In classical sketching algorithms, limited memory induces random collisions, meaning accuracy guarantees are typically probabilistic and the collision error probability shrinks with sketch size~\cite{nelson2011sketching}. In contrast, our approximate multiplication procedure is deterministic, with collision residual controlling the precision in the resulting compressed block encoding of the desired multiplication sequence. Here, ancilla count plays the role of ``sketch dimension'', with each additional ancilla exponentially reducing the collision residual. Overall, this analogy could suggests that classical sketching techniques might	 prove useful in further developing and tightening quantum space–accuracy tradeoffs---an exciting direction for future work.

\subsection{Background}
Although quantum computation is inherently unitary, quantum \emph{subsystems} can undergo non-unitary dynamics, if one allows post-selection. Thus, non-unitary linear operations on an $n$-qubit system can be represented by expanding the computational workspace, i.e., adding extra ``ancillary'' qubits, and treating the main $n$-qubit system as a subsystem of this larger workspace. 

While this procedure may sound abstract, decades of quantum algorithms research has formalized a simple, yet fundamental primitive for implementing such non-unitary embeddings in quantum computation\,---\,originally defined by~\cite{low2016HamSimQubitization} and subsequently termed by~\cite{chakraborty2018BlockMatrixPowers} as a ``block encoding''. 
Formally, a block encoding embeds a normalized, arbitrary (i.e., not necessary unitary) matrix $A$ into a larger unitary matrix,
\begin{align} \label{eqn:be}
    U_A = \begin{bmatrix}
        A & * \\
        * & *
    \end{bmatrix}, \quad \text{ where } A = (\bra{0^a} \otimes I)~U_A~(\ket{0^a} \otimes I).
\end{align}
More generally, the unitary $\widetilde{U}_A$ is called an $(\alpha, a, \epsilon)$-block encoding of $A$ if 
\begin{align}
   \| A - \alpha \cdot \bra{0^a} \widetilde{U}_A \ket{0^a} \| \leq \epsilon. 
\end{align}
Given an arbitrary $n$-qubit state $\ket{\psi}$, $A \ket{\psi}$ is implemented via the circuit that computes $U_A \ket{0^a}\ket{\psi}$, performs a measurment on a the first $a$ ancilla qubits, and post-selects on the $0^a$ outcome.

Despite its simplicity, block encoding has become one of the most important primitives in quantum algorithms~\cite{dalzell2023QuantumAlgSurvey}, largely thanks to qubitization~\cite{low2016HamSimQubitization} and quantum singular value transformation~\cite{gilyen2018QSingValTransf}, and their applications, e.g., in Hamiltonian simulation \cite{berry2015simulating, low2019Hamiltonian} and scientific computing. Indeed, block encoding methods have led to improved linear systems solvers~\cite{chakraborty2018BlockMatrixPowers} (with applications to regression), differential equation solvers \cite{fang2023time}, quantum information protocols~\cite{gilyen2020QAlgForPetzRecovery}, optimization~\cite{apeldoorn2019QAlgorithmsForZeroSumGames}, statistical estimation~\cite{gilyen2019DistPropTestQuant}, Gibbs sampling~\cite{chen2023QThermalStatePrep}, state preparation~\cite{lin2019OptimalQEigenstateFiltering}, and tomography~\cite{apeldoorn2022QTomographyWStatePrepUnis} algorithms. 
In practice, however, the implementation of block encodings can be resource intensive, requiring substantial computation time and ancilla space overhead \cite{clader2022quantum}.

Several techniques have also been established for manipulation of block encodings. For example, the linear combination of unitaries procedure \cite{childs2012hamiltonian} block encodes \emph{linear combinations} (i.e., addition) of general matrices. In this work, we specifically focus on procedures for \emph{multiplication} of a large number of block-encoded matrices. Beyond being an important subroutine for data embedding and processing~\cite{chakraborty2018BlockMatrixPowers}, multiplication of block encodings plays an important role in quantum simulation algorithms \cite{berry2015simulating, low2019Hamiltonian} and quantum differential equation solvers \cite{fang2023time}. Although multiplication of block encodings can be implemented with mid-circuit measurements (i.e. incoherent computation), such implementations are often challenging and limited in practice. For example, in superconducting architectures measurement times ($\sim100$ns-$1\mu$s) are typically much longer than gate times ($\sim 10$-$100$ns) \cite{krantz2019superconducting}, meaning repeated measurement could severely slow down the overall computation time. In neutral atom systems, imaging ejects atoms from the trap, which are time-consuming to reload \cite{bluvstein2024logical,manetsch2024tweezer}. Therefore, it is natural to consider \emph{coherent} implementations of multiplication of block encodings, via deferred measurement. Furthermore, coherent implementations have the important added benefits that they can be used as black-box subroutines in quantum algorithms and that amplitude amplification can be used to boost the probability of the successful measurement sequence to $\Omega(1)$.

\subsection{Related Work}
We build upon the compression gadet of \cite{low2019Hamiltonian} which, as we prove, uses an essentially optimal number of $\lceil\log_2K\rceil$ ancillas for \emph{exact} multiplication of $K$ block-encoded matrices. Let us illustrate the main idea in the special case of $K=2$. Suppose that 
\begin{align*}
	R=\begin{pmatrix} R_{00} & R_{0\perp}\\R_{\perp 0} & R_{\perp\perp}	\end{pmatrix} & \quad\text{ and }\quad
	C=\begin{pmatrix} C_{00} & C_{0\perp}\\C_{\perp 0} & C_{\perp\perp}	\end{pmatrix}
\end{align*}
are unitary block encodings of $A, B\in\mathbb{C}^{2^m\times 2^m}$ respectively, i.e., $A=R_{00}$ and $B=C_{00}$, and for simplicity both block encoding use $a$ ancilla qubits. While na\"{\i}ve multiplication of these block encodings would use $2a$ qubits~\cite{gilyen2018QSingValTransf}, the compression gadget requires only a single ancilla as follows
\begin{align*}
	(I_1\otimes R)(\text{NTOFF}_a\otimes I_m)(X\otimes C)=\!
	\begin{pmatrix} 
		R_{00}\!\! & R_{0\perp}\!\! & &\\
		R_{\perp 0}\!\! & R_{\perp\perp}\!\! & &\\
	& & R_{00}\!\! & R_{0\perp}\\
	& & R_{\perp 0}\!\! & R_{\perp\perp}
	\end{pmatrix}
	\begin{pmatrix} 
		C_{00}\!\! & C_{0\perp}\!\! & &\\
	& & C_{\perp 0}\!\! & C_{\perp\perp}\\
	& & C_{00}\!\! & C_{0\perp}\\
	C_{\perp 0}\!\! & C_{\perp\perp}\!\! & &
	\end{pmatrix}
	\!=
	\begin{pmatrix} 
		AB\!\! & *\\
		*&*
	\end{pmatrix},
\end{align*}
where the identities $I_k$ act on $k$ qubits and the $a$-qubit nil-controlled Toffoli is defined as $\text{NTOFF}_a=I_{a+1}+(X-I_1)\otimes\ketbra{0^a}{0^a}$.

This gadget plays a key role in other ancilla-reduction schemes as well, that are vaguely related to ours. Very recently, after the first version of this manuscript, \cite{rossi2026QElementwiseTransforms} showed how to exponentially reduce the ancilla requirements for implementing a block-encoding of $A^{\circ 2^d}$, whose $ij$ matrix element is defined as $(A_{ij})^{2^d}$. A na\"{\i}ve implementation would yield an ancilla overhead of $(a+n)\cdot 2^d$~\cite{guo2024QTransformerAcceleratingModel} when the $2^n \times 2^n$ matrix $A$ is provided via an $a$-ancilla unitary block encoding~$U_0$.

To explain the key idea, here we provide a direct circuit construction corresponding the their main result, which is not explicit in their work, as it is derived from composing abstract subroutines. Indeed, their recursive construction~\cite{rossi2026QElementwiseTransforms} can be understood via repeated squaring, i.e., a binary tree structure. Suppose we have a block-encoding $U_k$ of $A^{\circ 2^k}$ that uses $(a + k * (n+1))$ ancilla qubits. Then we can trivially get a block-encoding of both $A^{\circ 2^k} \otimes I_n$ and $I_n \otimes A^{\circ 2^k}$ using $(a + k * (n+1) + n)$ ancillas by simply using $U_k \otimes I_n$, where the $n$ qubit $I_n$ is inserted in the appropriate locations. We can take the product of these two block-encoded matrices using the above compression gadget resulting in an $(a + k * (n+1) + n + 1) = (a + (k+1) * (n+1))$ ancilla block-encoding of $A^{\circ 2^k} \otimes A^{\circ 2^k}$. As $A^{\circ 2^{(k+1)}}$ is the principal submatrix of $A^{\circ 2^k} \otimes A^{\circ 2^k}$ whose $i$-th row corresponds to the row $ii$, sandwiching this between $n$ CNOT gates results in a block-encoding of $A^{\circ 2^{(k+1)}}\!$, as described for example in ~\cite{guo2024QTransformerAcceleratingModel}. 

Combining this construction with careful control structures~\cite{rossi2026QElementwiseTransforms}, and standard linear combination of unitaries circuits~\cite{guo2024QTransformerAcceleratingModel}, one gets an $a+(n+2)\lceil\log_2(D)\rceil$ ancilla block encoding of the matrix whose matrix elements are transformed by an arbitrary (fixed) polynomial of degree $D$ whose coefficients absolute values sum to a number not exceeding $1$ and whose constant term is $0$.

\subsection{Our Contributions}
As discussed, practical implementation and manipulation of block-encodings can be resource intensive. This work focuses on reducing the number of ancillae necessary for block encoding operations, via two distinct approaches.

\paragraph{\textbf{Part I} (\Cref{sec:intro_part_i} and \Cref{sec:part_i}).} Here we assume black-box access to a block encoding $U_A$ implemented with many ancillae. We give a protocol for coherently uncomputing all but one of a $U_A$'s ancillae, effectively establishing a ``space-time tradeoff''. We show that for any matrix $A$ (s.t.\ $\|A\| <0.99$) and error $\epsilon$, one can get an $\epsilon$-precise single-ancilla block encoding of $A$ with $\calO(\log(1/\epsilon))$ uses of $U_A^{\pm1}$. 

It is important to note that the procedure does not strictly eliminate the $a$ original ancillae in implementing the block encoding, since the block encoding is itself queried in our procedure. However, the procedure  (approximately) returns $a-1$ of the block encoding's ancillae back to the $\ket{0^{a-1}}$ state. These can, thus, be treated as ``fresh'' ancillae for reuse in other tasks later in the algorithm (without intermediate measurements or classical post-selection). This is why, as previously discussed, our procedure can be leveraged by a quantum algorithm 
to reduce the overall ancilla-overhead, especially in the case of deep compositions of block-encodings. We conclude the section with discussion on the connections between our ancilla-uncomputation procedure and the  ``phase correction'' procedure for modular quantum signal processing~\cite{rossi2022multivariableQSP}.

\paragraph{\textbf{Part II} (\Cref{sec:intro_part2} and \Cref{sec:part_ii}).} Here we focus specifically on the ancilla-overhead of multiplication of block encodings. We prove that \emph{exact} multiplication of $K$ block encodings requires at least $\lceil \log_2 K \rceil$ ancillae\,---\,a lower-bound saturated by a slight modification to the compression gadet of \cite{low2019Hamiltonian}. Moreover, by relaxing from exact to \emph{approximate} multiplication of block encodings we show that, in certain regimes (where each block encoding is structured), only $a=\calO(1)$ ancillae are necessary to implement an $\calO(1/K^{2^a})$-precise multiplication of $K$ block encodings. This offers a substantial saving over the $a=\Theta(\log_2 K)$ ancillae required by the \cite{low2019Hamiltonian} compression gadget for exact-multiplication. While the result assumes structure about the block encodings (and is therefore not fully generic), it serves as an exciting proof-of-concept that ancilla overhead can in fact be reduced below the logarithmic lower-bound necessary for exact multiplication. We further motivate how the guarantees of this procedure are relevant in the settings of Hamiltonian simulation and quantum differential equation solvers. Finally, we show that oblivious amplitude amplification can be used to boost the success probability for $\epsilon$-approximate multiplication to $1-\calO(\epsilon^2)$.

\section{Part I: Approximate Single-Ancilla Block Encodings} \label{sec:intro_part_i}

As discussed in the introduction, implementation of block encodings can often be resource intensive, especially with respect to ancilla overhead. Given an ``inefficient'' block encoding (i.e. a block encoding requiring more than one ancilla qubit) of matrix $\general$, with $\|\general\|<1$, we prove that there exists a simple and efficient quantum algorithm for coherent uncomputation of all but one ancilla. This protocol achieves an $\epsilon$-precise (Hermitian) single-ancilla block encoding of $\general$, returning the remaining ancilla back to the $\ket{0}$ state and freeing them for reuse for other tasks in a quantum algorithm. Overall, this can be interpreted as a ``space-time tradeoff'', in which ancilla space overhead is traded for additional computation time. 

Formally, we prove the following main theorem regarding ancilla uncomputation of block encodings of Hermitian matrices $\herm$.

\begin{restatable}[Ancilla Uncomputation for Hermitian Matrices]{theorem}{mainUncompute}\label{thm:main_uncompute}
    Let $V_{\herm}$ be a $(1, a, 0)$-block-encoding of Hermitian matrix $\herm$, where $\norm{\herm} \leq 1-\delta$, for some $\delta \in (0,1)$. For any $\eps \in (0,1)$, there exists a quantum circuit, requiring a small constant number $c=\calO(1)$ of additional ancillae and $\calO\left(\frac{1}{\delta}\log(\frac{1}{\epsilon})\right)$-queries to $V_\herm$ and $V_\herm^\dagger$, that implements a $(1,1,\epsilon)$-block encoding of the canonical block encoding $U_\herm$ from \Cref{eqn:u_herm}; in turn this block encoding circuit is also a $(1,1,\epsilon)$-block encoding of $\herm$ which returns all but the last ancilla to the $\ket{0}$ state (up to error $\epsilon$).
\end{restatable}
\noindent Although \Cref{thm:main_uncompute} is only proven for block encodings of Hermitian matrices $\herm$, in \Cref{sec:non_hermitian} we describe how the procedure can be generalized to non-Hermitian sub-normalized matrices $\general$, resulting in the following corollary.
\begin{restatable}[Ancilla Uncomputation for General Matrices]{corollary}{generalCorollary}\label{thm:general_corollary}
     Let $V_{\general}$ be a $(1, a, 0)$-block-encoding of general matrix $\general$, where $\norm{\general} \leq 1-\delta$, for some $\delta \in (0,1)$. For any $\eps \in (0,1)$, there exists a quantum circuit, requiring a small constant number $c=\calO(1)$ of additional ancillae and $\calO\left(\frac{1}{\delta}\log(\frac{1}{\epsilon})\right)$-queries to $V_\general$ and $V_\general^\dagger$, 
     that implements a $(1,1,\epsilon)$-block encoding of the canonical block encoding $U_\general$ from \Cref{eqn:u_gen}; in turn this block encoding circuit is also a $(1,1,\epsilon)$-block encoding of $\general$ which returns all but the last ancilla to the $\ket{0}$ state (up to error $\epsilon$).
\end{restatable}

To motivate the proof approach, let us first consider standard ways in which one might try to achieve a single-ancilla block-encoding of an arbitrary sub-normalized matrix $\general$. The simplest procedure would be to simply measure $a-1$ of $V_\general$'s $a$ ancilla and post-select for the $\ket{0^{a-1}}$ outcome. The key downside is that this procedure is \emph{incoherent}, i.e., multiple circuit runs and measurements will be required to achieve the  $\ket{0^{a-1}}$ outcome.  Alternatively, one could try to avoid post-selection by using amplitude amplification to boost the amplitude of the $\ket{0^{a-1}}$ ancilla state. However, this requires expensive reflections about the initial state, except for the special case when we work with block encodings of \emph{unitary} matrices (or isometries) where oblivious amplitude amplification works. Thus, it is not immediately clear how to coherently uncompute the $a-1$ ancilla.

\begin{algorithm}[t] \small
    \caption{\small Approximate Hermitian Block Encoding with One Ancilla} \label{alg:uncompute}

    \vspace{0.05in}
    \textbf{Inputs:} $(1,a,0)$-block encoding $V_\herm$ of Hermitian matrix $\herm$ ($\|H\|\leq 1-\delta$) and error $\epsilon \in (0,1)$ 

    \vspace{0.05in}
    \textbf{Output:} $(1,1,\epsilon)$-block encoding $\widetilde{U}_\herm$ of $\herm$ 

    \vspace{0.05in}
    \textbf{Algorithm Pseudocode:}
    \begin{enumerate}
        \item Define the subroutine that computes unitary $V_{\frac{I-\herm^2}{2}}$, which is a $(1,a+1,0)$-block-encoding of the matrix $\frac{I-\herm^2}{2}$. This is implemented via linear combination of unitaries (LCU) of $V_{\herm}$.
        \item Define the subroutine that computes unitary $V_{\frac{\sqrt{I-\herm^2}}{\sqrt{8}}}$, which is a $(1,a+\calO(1),\frac{\eps}{9})$-block encoding of $\frac{\sqrt{I-\herm^2}}{\sqrt{8}}$. This is implemented by applying quantum singular value transformation (QSVT) to $V_{\frac{I-\herm^2}{2}}$, with a degree-$d$ polynomial approximation of the transform function $f(x) := \frac{1}{2} \sqrt{x}$. The QSVT requires $\calO(1)$ ancillae and $d=\calO(\frac{1}{\delta}\log(\frac{1}{\epsilon}))$ queries to $V_\herm$ and $V_\herm^\dagger$ to achieve the desired precision.
        \item Use LCU and previous subroutines to implement $W_{U_\herm} = \sin(\pi/14) \left( \sqrt{8} \cdot X \otimes V_{\frac{\sqrt{I-\herm^2}}{\sqrt{8}}}  +  Z \otimes V_{\herm}  \right)$,
        which is a $(1,a+\calO(1),\frac{\eps}{14})$-block-encoding of $\sin(\pi/14) \cdot U_{\herm}$.
        \item Perform oblivious amplitude amplification (OAA) to map $W_{U_\herm}$ to unitary $\widetilde{W}_{U_\herm}$, which is an $(1,a+c, \epsilon)$-block encoding of $U_\herm$, for $c=\calO(1)$.
    \end{enumerate} 
\end{algorithm}

The proof of \Cref{thm:main_uncompute} fundamentally leverages the key insight that oblivious amplitude amplification holds for any block encoding of a \emph{unitary} matrix. Thus, in our proposed protocol (\Cref{alg:uncompute}), we offer a quantum algorithm---involving strategic application of linear combination of unitaries \cite{childs2012hamiltonian} and the quantum singular value transformation \cite{gilyen2018QSingValTransf}---that implements $\widetilde{W}_{U_\herm}$, an $\eps$-precise block encoding  of the \emph{canonical unitary block encoding} matrix
\begin{align}\label{eqn:u_herm}
    U_{\herm} = Z \otimes \herm + X \otimes \sqrt{I - \herm^2}= \begin{pmatrix}
    \herm & \sqrt{I - \herm^2} \\
    \sqrt{I - \herm^2} & -\herm
\end{pmatrix},
\end{align}
which is itself a Hermitian single-ancilla block-encoding of $\herm$.  Since $U_{\herm}$ is unitary, oblivious amplitude amplification can be used to boost the probability of $W_{U_\herm}\!$'s corresponding ancilla measurement outcome to $\Omega(1)$. Thus, all of $W_{U_\herm}\!$'s ancilla are reset to the $\ket{0}$ state and when the ancillas are initially set to $\ket{0}$ the action on the subsystem is $\epsilon$-close to $U_\herm$. These guarantees for the above algorithm are proved in \Cref{sec:part_i}.

In order to generalize this procedure (\Cref{thm:main_uncompute}) for general non-Hermitian matrices (\Cref{thm:general_corollary}) the critical change is in the form of the unitary block encoding. Namely, for general sub-normalized matrix $\general$, the protocol will implement an $\eps$-precise block encoding of the \emph{unitary} matrix
\begin{align}\label{eqn:u_gen}
    U_\general = \begin{pmatrix}
    \sqrt{I - \general^\dagger \general} & \general^\dagger \\
    \general & -\sqrt{I - \general \general^\dagger}
    \end{pmatrix},
\end{align}
which is itself a single-qubit Hermitian unitary block-encoding of $\general$.\footnote{Note that the block-encoding is of the form $\general = \bra{1}U_\general \ket{0}$.} Since $U_\general$ is Hermitian and unitary like $U_\herm$, the protocol for the Hermitian case can be lifted to this setting.

Overall, this result is perhaps somewhat surprising, as it states that---for any matrix $\general$ (with $\| \general \| \leq 1-\delta$), $\calO(t)$ queries to \emph{any bock encoding} of $\general$ and $\general^\dagger$, and $\calO(1)$ intermediate computation ancilla---there is a simple quantum circuit implementing an $e^{-\calO(\delta t)}$-precise single-ancilla block encoding of $\general$.

\subsection{Connections to Modular Quantum Signal Processing}
We forsee the main application of this ancilla uncomputation procedure in algorithms with deep, nested linear algebra operations, freeing-up and reusing ancillas utilized in composite operations. While it might not be easy to see immediately, our results are also deeply connected the problem of ``phase correction'' in modular quantum signal processing.

Quantum signal processing (QSP) \cite{low2016HamSimQSignProc}, the predecessor of QSVT, is an important framework for implementing polynomial transformations of unitary operator eigenvalues, via sequences of controlled single-qubit rotations. 
Formally, for the scalar $x\in[-1,1]$ with block encoding
\begin{align}
    U(x) = \begin{pmatrix}
        x & \sqrt{1-x^2} \\
        \sqrt{1-x^2} & -x
    \end{pmatrix},
\end{align}
rotation matrix $O(x)=U(x) \sigma_z$, and a specially chosen set of phase factors $\Phi := (\phi_0,\cdots, \phi_d)\in\mathbb{R}^{d+1}$, a QSP protocol implementing polynomial transform $P(x)$ is:
\begin{align}
    U_\Phi(x) = e^{i \phi_0 \sigma_z} \prod_{j=1}^d O(x)e^{i\phi_j \sigma_z} = \begin{pmatrix}
        P(x) & -Q(x)\sqrt{1-x^2} \\
        Q^*(x)\sqrt{1-x^2} & -P^*(x)
    \end{pmatrix}.
\end{align}
QSP is a special instance of the more general multivariable quantum signal processing (M-QSP) framework, introduced by \cite{rossi2022multivariableQSP}. For $t\in \mathbb{N}$, a $t$-variable M-QSP protocol specified by phase factors $\Phi := (\phi_0,\cdots, \phi_n\}\in\mathbb{R}^{n+1}$ and string $s \in [t]^{\times n}$ implements the unitary circuit
\begin{align}
    \Phi[U_0, \cdots, U_{t-1}] \equiv e^{i \phi_0 \sigma_z} \prod_{j=1}^n U_{s_j}e^{i\phi_j Z},
\end{align}
where $U_0,\cdots, U_{t-1}$ are unitaries of the form $U_i \equiv e^{i\theta_i \sigma_x}$ for unrelated $\theta_i$. In this more general M-QSP framework, a single-variable QSP protocol can be expressed as $\Phi[U]=e^{i \phi_0 Z} \prod_{j=1}^d U e^{i\phi_j Z}$.

A natural question is whether polynomials can be \emph{composed} within the QSP and M-QSP frameworks, enabling modular quantum signal processing. To this end, \cite{rossi2023semanticEmbedding} showed that anti-symmetric\footnote{Anti-symmetric QSP protocols have phase factors of the form: $\Phi = (\phi_0,\phi_1,\cdots, \phi_d,-\phi_d,\cdots, -\phi_1,-\phi_0)$ or $\Phi = (\phi_0,\phi_1,\cdots, \phi_d,0,-\phi_d,\cdots, -\phi_1,-\phi_0)$.} single-variable QSP protocols can be embedded within themselves. Concretely, for anti-symmetric QSP protocols $\{\Phi_0,\cdots,\Phi_n\}$ implementing polynomials $\{P_0,P_1,\cdots,P_n\}$, the composition of polynomials $P_n \circ \cdots \circ P_1 \circ P_0 (x)$ is implemented via the QSP protocol composing $\Phi_{k+1}[\Phi_k]$ recursively. This recursion is possible because anti-symmetric QSP protocols are invariant under \emph{twisting}---i.e. for any unitary $U$ and rotation $e^{i\phi \sigma_z}$, $\Phi[e^{i\phi \sigma_z} U e^{-i\phi \sigma_z}]= e^{i\phi \sigma_z} \Phi[U] e^{-i\phi \sigma_z}$. Thus, accumulating $\sigma_z$ conjugation terms have no affect on the implemented polynomial:
\begin{align}
    P(x) = \bra{0} \Phi[U] \ket{0} = \bra{0} e^{i\phi \sigma_z}\Phi[U] e^{-i\phi \sigma_z}\ket{0}=\bra{0} \Phi[e^{i\phi \sigma_z}U e^{-i\phi \sigma_z}]\ket{0}.
\end{align}

In the multi-variable setting, however, each anti-symmetric protocol $\Phi_i$ is conjugated by a distinct $e^{i\phi_i \sigma_z}$ term, meaning it is not possible to simultaneously factor all phase terms out of the overall M-QSP protocol. To address this, \cite{rossi2023modularQSP} introduced a technique for suppressing disagreeing $\sigma_z$-conjugations, thereby preventing the conjugations from interfering with the outer M-QSP protocol. Specifically, they offer a procedure \cite[Lemma II.1]{rossi2023modularQSP} for mapping so-called ``twisted embeddable'' unitaries, matrices of the form 
\begin{align}
    \widetilde{U}(\theta,\varphi) = e^{i \varphi \sigma_z/2}e^{i \theta \sigma_x}e^{-i \varphi \sigma_z/2} = \begin{pmatrix}
        \cos(\theta) & i \sin(\theta) e^{i \varphi} \\
        i \sin(\theta) e^{-i \varphi} & \cos(\theta)
    \end{pmatrix},
\end{align}
with unknown $\theta$ and $\varphi$, to their embeddable component 
\begin{align}
    e^{i \theta \sigma_x}= \cos(\theta) \cdot I + i \sin(\theta) \cdot \sigma_x=\begin{pmatrix}
        \cos(\theta) & i \sin(\theta) \\
        i \sin(\theta) & \cos(\theta)
    \end{pmatrix}.
\end{align}
Thus, they align the axes of rotation for two or more unitary oracles in an efficient, black-box way.

We will now show how our ancilla uncomputation procedure can also straightforwardly be used to correct these twisted embeddable unitaries. Namely, in \Cref{thm:main_uncompute}, we offer a procedure that, with black-box queries to block encoding $V_\herm$, implements the matrix 
\begin{align}
    U_{\herm} = Z \otimes \herm + X \otimes \sqrt{I - \herm^2} = \begin{pmatrix}
            \herm & \sqrt{I - \herm^2} \\
            \sqrt{I - \herm^2} & -\herm
        \end{pmatrix}.
\end{align}
 Since the matrix $\widetilde{U}(\theta,\varphi)$ is itself a block encoding of $\herm=\cos(\theta)$, by setting $V_\herm=\widetilde{U}(\theta,\varphi)$, the procedure implements the matrix
\begin{align}
    U_{\cos\theta} = \cos \theta \cdot Z + \sin \theta \cdot X = \begin{pmatrix}
        \cos\theta & \sin\theta \\
        \sin\theta & -\cos\theta
    \end{pmatrix}.
\end{align}
Conjugation of this matrix with the $S=\sqrt{Z}$ gate  achieves the desired embeddable component:
\begin{align}
    S \cdot U_{\cos\theta} \cdot S = \cos \theta \cdot S Z S + \sin \theta \cdot S X S = \cos \theta \cdot I + \sin \theta \cdot iX= \begin{pmatrix}
        \cos(\theta) & i\sin(\theta) \\
        i\sin(\theta) & \cos(\theta)
    \end{pmatrix} = e^{i \theta \sigma_x}.
\end{align}

\section{Part II: Approximate Multiplication of Block Encodings} \label{sec:intro_part2}
In the second half of the work, we move away from implementation of block encodings to their manipulation. Specifically, we focus on procedures for \emph{multiplication} of a large number of block-encoded matrices which, as described in the introduction, plays an important role in quantum simulation algorithms \cite{berry2015simulating, low2019Hamiltonian} and quantum differential equation solvers \cite{fang2023time}.

Given a sequence of block encoding unitaries $\{U_{A_i}\}_{i=1}^K$, multiplication of their block-encoded matrices, i.e. $A_K \cdots A_2 A_1 \ket{\psi}$, is implemented \emph{incoherently} via the circuit depicted in \Cref{fig:inter_meas_setup}. In this implementation, $U_{A_{i+1}}$ is applied only if the ancilla measurement outcome following $U_{A_i}$ is $0^a$ (indicating successful application of $A_i$). For any non-$0^a$ measurement outcome ($\perp$), we discard the entire computation and start again. This incoherent implementation requires intermediate measurements and interfacing the quantum system with classical controls which, as discussed in the introduction, is challenging to effectively implement in practice. Therefore, it is natural to consider \emph{coherent} implementations of multiplication of block encodings, via deferred measurement.

Crucially, in a coherent deferred measurement scheme, we still must identify whether all $K$ measurements measured to the desired $0^a$ outcome or if any measured to an undesired perpendicular state $\perp$. A naive coherent implementation would apply a C$_{\Pi_\perp}$NOT\footnote{Here $C_{\Pi_\perp}$NOT denotes an $X$ gate conditioned on the $\Pi_\perp=\ketbra{\perp}{\perp}=I-\ketbra{0^a}{0^a}$ state. That is, the $X$ gate is applied for any non-$\ket{0^a}$ state.} from every qubit in the ancilla register to a target ancilla measurement qubit, for each of the first $K-1$ unitaries, as depicted in \Cref{fig:def_meas_naive}. In total, this requires $K-1$ ``measurement ancilla''. 

\begin{figure}[t!]
    \centering
    \includegraphics[width=.9\textwidth]{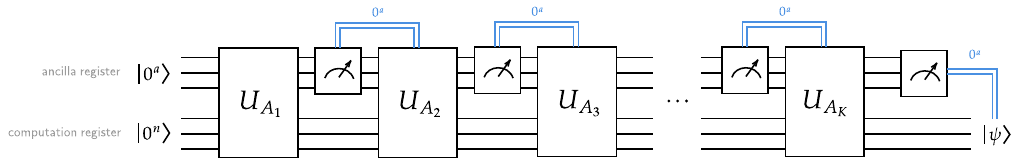}
    \caption{Multiplication of $K$ block encodings implemented with repeated incoherent measurements. Each incoherent measurement is only accepted if the outcome is $0^a$. If a measurement is rejected the entire computation must be discarded and restarted.}
    \label{fig:inter_meas_setup}
    \vspace{0.05in}
    \includegraphics[width=.85\textwidth]{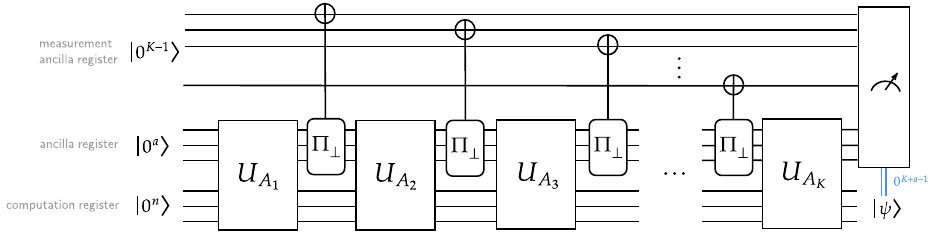}
    \caption{A naive implementation of multiplication of $K$ block encodings consisting only of coherent measurements. We refer to the implementation as ``naive'' since it requires $K-1$ ancilla qubits, which can be substantially reduced.}
    \label{fig:def_meas_naive}
\end{figure}

However, \cite{low2019Hamiltonian} showed that, via a more sophisticated procedure (replacing the C$_{\Pi_\perp}$NOT gates with C$_{\Pi_\perp}$ADD gates),  known as the ``compression gadget'', coherent multiplication of block encodings is implementable with only $m=\lceil \log_2 K\rceil$  measurement ancilla.\footnote{Note that the actual compression gadget of \cite{low2019Hamiltonian} required $\lceil \log_2 (K+1)\rceil$ ancillae, since they did not leverage the fact that the measurement of the ancilla register can be used to determine success of the $K^\text{th}$ block encoding.} (The \cite{low2019Hamiltonian} compression gadget is a special instance of the circuit depicted in \Cref{fig:add_gadget}, with $p=\lceil \log_2 K\rceil$.) Since this is the most ancilla-efficient proposal for coherent multiplication of block encodings known to-date, the key question we explore in Part II is:
\begin{quote} 
   \textit{Can multiplication of $K$ block encodings be achieved with less than $\lceil \log_2 K\rceil$ ancillae?}
\end{quote}
To this end, our work makes several contributions:
\begin{enumerate}
    \item We formalize the \circclass~circuit class for coherent multiplication of block encodings and we offer a simple universal implementation for all such circuits. 
    \item We prove a $\lceil \log_2 K \rceil$ measurement ancillae lower-bound for \circclass~circuits implementing \emph{exact} multiplication of $K$ block encodings. Thus, the \cite{low2019Hamiltonian}  compression gadget (\Cref{fig:add_gadget}) is optimal for exact multiplication.  
    \item We propose $\epsilon$-\emph{approximate} compression gadgets, which implement $\epsilon$-\emph{approximate} multiplication of block encodings, i.e. an $\epsilon$-precise block encoding of the desired multiplication sequence. 
    \item We show how oblivious amplitude amplification can be used to boost the overlap between an $\epsilon$-approximate multiplication circuit output and the desired output to $\Omega(1-\epsilon^2)$.
    \item We propose the $p$-Modular Addition Compression Gadget ($p$-MACG)---a generalization (parameterized by $p$) of the \cite{low2019Hamiltonian} compression gadget,  where addition is performed mod-$2^p$ instead of mod-$2^{\lceil \log_2 K \rceil}$ (\Cref{fig:add_gadget}). 
    \item For any sequence of $K$ block encodings with bounded distance from identity, we prove that the $p$-MACG with $p=\calO(1)$ measurement ancillae is an $\calO(1/K^{2^p})$-compression gadget. Thus, with just one ancilla, the $1$-MACG achieves $\calO(1/K^{2})$-approximate multiplication. This demonstrates that approximate multiplication can surpass the exact multiplication ancilla lower-bound, while still achieving high-precision.
    \item We describe how multiplication of block encodings with bounded distance from the identity arises naturally in applications such as Hamiltonian simulation and quantum differential equation solvers, motivating the practical relevance of the $p$-MACG. 
\end{enumerate}
We will now briefly outline each of these contributions, but for full details we refer the reader to \Cref{sec:part_ii}.

\begin{figure}[t!]
    \centering
    \includegraphics[width=.75\textwidth]{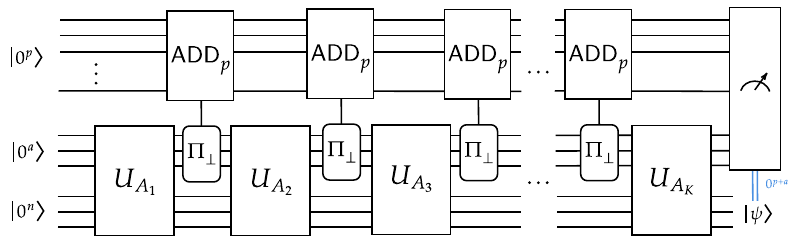}
    \caption{The proposed $p$-Modular Addition Compression Gadget, where $\addp \ket{x} = \ket{x+1 \bmod{2^p}}$.  The \cite{low2019Hamiltonian} compression gadget is recovered for $p=\lceil\log_2 K \rceil$.}
    \label{fig:add_gadget}
\end{figure}

\subsection{An Ancilla Lower-Bound for Exact Multiplication of Block Encodings}
The first main contribution of this work is an ancilla lower-bound for, what we refer to as, \emph{exact} multiplication of block encodings (EMBE). 
\begin{definition}[EMBE]
    A circuit achieves exact multiplication of block encodings (EMBE) for the sequence of block encodings $(U_{A_i})_{i=1}^K$ if it implements a unitary $U_\embe$ satisfying
    \begin{align} 
        \bra{0^{m+a}} U_\embe \ket{0^{m+a}} = A_K \cdots A_2 A_1
    \end{align}
\end{definition}

In order to prove the lower-bound, we begin by formalizing the class of circuits implementing multiplication of block encodings. Suppose that we have a sequence of $K$ block encodings $(U_{A_i})_{i=1}^K$ such that, with $n$ qubits in the computational register, and $a$ ancillae for the block encoding. We will assume that the sequence of all ``good'' measurements is governed by projectors $(\pzeroa)_{i=1}^K$
that we want to post-select for. As we saw in \Cref{sec:intro}, there are multiple possible circuits that coherently implement this multiplication of block encodings (especially consider a varying number of measurement ancillae $m$). We will, thus, refer to the class containing all such circuits as the Multiple Coherent Measurement (\circclass) circuit class. We further prove that the following is a simple universal circuit parameterization for \circclass~circuits.

\begin{definition}[\circclass$_{\vec{U}_A, \vec{\Pi}}$ Circuit Class] 
    Suppose we are given series of $K$ $(n+a)$-qubit block-encodings $\vec{U}_A = (U_{A_i})_{i=1}^K$. Without loss of generality, assume that all $U_{A_i}$ have $\pzeroa$ as the $a$-qubit projector onto ``good'' measurement outcome and $\pperp=I_a-\pzeroa$ as the $a$-qubit projector onto ``bad'' measurement outcome. With this, we define $\circclass_{\vec{U}_A, K,m}$ to be the class of all circuits of the form,
    \begin{align}
        \circclass_{\vec{U}_A, K,m} (\vec{V}, Q) = (Q\otimes U_K) \cdot \prod_{i=1}^{K-1} \left(  \left(V_{K-i}\otimes \pperp \otimes I_n \right) \left(I_{m+a} \otimes U_{K-i}\right)\right),
    \end{align}
    parameterized by the $m$-qubit unitaries $Q$ and $\vec{V} = (V_i)_{i=1}^{K-1}$. The circuit is depicted in \Cref{fig:lemma_con_uni}.
\end{definition}
\begin{figure}[t!]
    \centering
    \includegraphics[width=.8\textwidth]{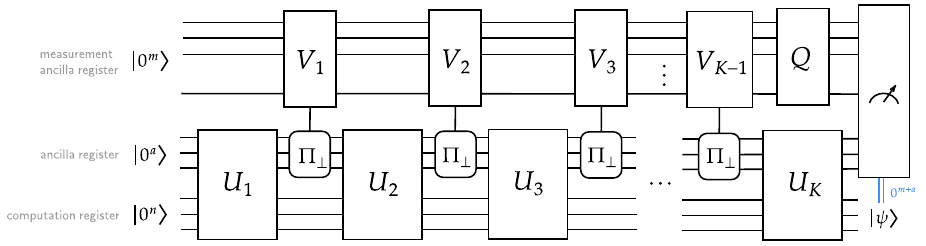}
    \caption{Universal coherent implementation of a $\circclass_{\vec{U}_A, K,m} (\vec{V}, Q)$ circuit.}
    \label{fig:lemma_con_uni}
\end{figure}
\noindent Leveraging the structure of the \circclass~circuit class and a Hilbert-space dimension counting argument, we prove the following measurement ancilla lower-bound for multiplication of block encodings.
\begin{theorem}[EMBE Ancilla Lower-Bound] 
    In order to coherently implement an exact multiplication of $K$ block encodings, at least $m=\lceil\log_2 K\rceil$ measurement ancillae are required.
\end{theorem}
\noindent Since the previously described \cite{low2019Hamiltonian} compression gadget saturates this bound, it is therefore proven to be  ancilla-optimal for EMBE.

\subsection{Approximate Multiplication of Block Encodings}
In an attempt to surpass the ancilla lower-bound for \emph{exact} multiplication of block encodings (EMBE), we propose the notion of \emph{approximate} multiplication of block encodings (AMBE).
\begin{definition}[$\eps$-AMBE]
    A circuit achieves $\epsilon$-approximate multiplication of block encodings ($\epsilon$-AMBE) for the sequence of block encodings $(U_{A_i})_{i=1}^K$ if it implements a unitary $U_\ambe$ satisfying
    \begin{align} 
        \|A_K \cdots A_2 A_1-\bra{0^{m+a}} U_\ambe \ket{0^{m+a}}\| \leq \epsilon.
    \end{align}
\end{definition}
\noindent In other words, a circuit implementing $\eps$-AMBE implements an $\eps$-precise block encoding of the desired multiplication matrix $A_K \cdots A_2 A_1$.

Since the logarithmic overhead of EMBE circuits is most problematic in the large-$K$ limit, we focus on reducing the ancilla overhead in near-asymptotic regimes. In order for multiplication of block-encodings to have high enough success probability to be useful in near-asymptotic regimes, the block encodings themselves must have high success probability. Therefore, we specifically focus on regimes in which each block encoding has $\calO(1/K)$-bounded deviation from the identity, i.e.
\begin{align} \label{eqn:constraint}
    \| U_{A_i} -  I\| \leq \eta_\text{max} = \calO(1/K).
\end{align}
Although this constraint might seem strong, in \Cref{sec:applications} we argue that it is well-motivated by practical quantum algorithms, such as Hamiltonian simulation and differential equation solvers.

In this regime, to surpass the EMBE ancilla lower-bound, we propose the $p$-Modular Addition Compression Gadget ($p$-MACG), as depicted in \Cref{fig:add_gadget}. This is a generalization of the \cite{low2019Hamiltonian} compression gadget which, instead of acting on $m=\lceil \log_2 K \rceil$ measurement ancilla, parametrizes the number of measurement ancilla $m$ by $p=\calO(1)$ and thus performs addition mod-$2^p$. Similar to the \cite{low2019Hamiltonian} compression gadget, the $p$-MACG always correctly identifies the good measurement sequence $\{0^a\}^K$. However, since $2^p<K$, there are certain bad measurement sequences with $2^p$ total bad measurements that the $p$-MACG misclassifies as the good measurement sequence. This incurs an error in the block encoding. However, we prove the following key result, bounding the total error as a function of $p$.
\begin{theorem}[$p$-MACG Error Bound]
    Let $(U_{A_i})_{i=1}^K$ be a sequence of $K$ block encodings, with maximal deviation coefficent $\eta_{\max}=\frac{c}{K}$ for some constant $c\in \mathbb{R}$.
    For any $p\in \mathbb{Z}$ such that $p=\calO(1)$, the  $p$-ACG implements an $\epsilon$-precise block encoding of the desired multiplication sequence $A_K \cdots A_2 A_1$, with inverse-polynomial precision in $K$,
    \begin{align}
        \epsilon = 2 e^c \cdot \left(\frac{e \cdot c^2 }{K\cdot 2^p}\right)^{2^p} = \calO\left(\frac{1}{K^{2^p}}\right).
    \end{align}
\end{theorem}

Therefore, in the regime of block encodings with $\calO(1/K)$-bounded operator norm distance from the identity, we prove that the $p$-MACG is an $\calO(K^{-2^p})$-AMBE. This implies that even a $1$-MACG, using only a single measurement ancilla, can achieve an $\calO(K^{-2})$-precise block encoding of $A_K \cdots A_2 A_1$. Thus, the EMBE ancilla lower-bound can be surpassed in the approximate multiplication setting.

\subsection{Practical Applications of the \texorpdfstring{$p$}{p}-MACG} \label{sec:applications}
As previously discussed, the guarantees of the $p$-MACG only hold in regimes where all $K$ block encodings have $\calO(1/K)$-bounded distance from the identity (\Cref{eqn:constraint}). We will now motivate this constraint by demonstrating how it naturally emerges in the settings of Hamiltonian simulation and quantum differential equation solvers.

In Hamiltonian simulation, the goal is to implement the evolution operator of some Hamiltonian $H$ at a specified time $t$, i.e. $U(t)=e^{-iHt}$. However, in practice, direct simulation of the full evolution is often intractable and approximations are necessary. A common approach is Trotterization, in which the Hamiltonian is decomposed into a sum of simpler terms, i.e. $H=\sum_i H_i$, to which the Trotter approximation $e^{\delta (H_i+H_j)}\approx e^{\delta H_i}e^{\delta H_j}$ can be applied. Thus, the evolution operator can be expressed as a product of short time-scale  evolution terms $e^{-iH_i t/K}$, as
\begin{align}
    U(t)= e^{-iHt} = \left(e^{-i(\sum_i H_i)t/K}\right)^K \approx \left(\prod_i e^{-iH_i t/K}\right)^K.
\end{align}
Expanding these small evolution terms via a Taylor series gives
\begin{align}
    e^{-iH_i t/K} = I - i \frac{t}{K} H_i - i \frac{t^2}{2K} H_i^2 + \cdots.
\end{align}
Under the standard assumption that $\|H_i\|\leq 1$, for a fixed time $t$ and sufficiently large $K$, the dominant deviation from the identity comes from the first term, i.e.
\begin{align}
    \|I - e^{-iH_i t/K}\| = \calO\left(1/K\right).
\end{align}
All that remains is to ensure the block encoding of $e^{-iH_i t/K}$ also has $\calO\left(1/K\right)$-bounded distance from identity. Theoretically, this can be straightfowardly achieved using the block encoding 
\begin{align}
    U_{e^{-iH_i t/K}}=\ketbra{0}{0}\otimes e^{-iH_i t/K} + \ketbra{1}{1}\otimes I,
\end{align}
which has the desired distance from the identity:
\begin{align}
    \|I - U_{e^{-iH_i t/K}}\| = \left\|\ketbra{0}{0}\otimes (I-e^{-iH_i t/K}) + \ketbra{1}{1}\otimes (I-I)\right\| = \|I - e^{-iH_i t/K}\| = \calO\left(1/K\right). 
\end{align}
While such a block encoding might be practically challenging to implement, it demonstrates that block encodings satisfying the $p$-MACG criterion exist for Hamiltonian simulation in large-$K$ regimes.

In the setting of quantum differential equation solvers, specifically for linear ordinary differential equations, the goal is to solve the initial value problem 
\begin{align}
    \frac{d}{dt} \ket{\psi(t)} = A(t) \ket{\psi(t)}, \quad \ket{\psi(0)}=\ket{\psi_0}
\end{align}
for some $A(t)$, a matrix-valued function of the time variable $t$, to recover the state $\ket{\psi(T)}$ at desired time $t=T$. This can be directly related to the problem of time-dependent Hamiltonian simulation by setting $A(t)=-iH(t)$ for some Hamiltonian $H(t)$. For appropriately bounded $A$, the unique solution to the time-dependent differential equation is given by the Dyson evolution operator
\begin{align}
    \Xi = \calD(t,t_0) &= \calT \exp\left(-i\int_{t_0}^t A(s) ds\right) = \sum_{k=0}^\infty (-i)^k \int_{t_0}^t d s_k \int_{t_0}^{s_k} d s_{k-1} \cdots \int_{t_0}^{s_2} d s_1 A(s_k) \cdots A(s_1)
\end{align}
where $\calT$ is the time-ordering operator. A common approximation approach is the time-marching strategy, as described in \cite{fang2023time}, in which the full time $T$ is divided into shorter time intervals for which the system is solved and then the information is propagated. In the case that these intervals are choosen uniformly, i.e. $\Delta t = t/K$, the evolution operator is decomposed as
\begin{align}
    \Xi = \calD(t,0) \approx \calD(K\Delta t,(K-1)\Delta t) \cdots\calD(2\Delta t,\Delta t) \cdot \calD(\Delta t,0) = \Xi_K \Xi_{K-1} \cdots \Xi_2 \Xi_1,
\end{align}
where $\Xi_j$ are themselves short time-scale Dyson evolution operators,
\begin{align}
    \Xi_j = \calD(t_j+\Delta t, t_j) = I - i \int_{t_j}^{t_j+\Delta t} d s_1 A(s_1) - \int_{t_j}^{t_j+\Delta t} d s_2 \int_{t_j}^{s_2} d s_1 A(s_2)A(s_1) + \cdots.
\end{align}
If $A(t)$ is uniformly bounded, i.e. $\|A(t)\| \leq \lambda$ for all time $t$, then the terms of $\Xi_j$ are bounded as
\begin{align}
    \left\| \int_{t_j}^{t_j+\Delta t} d s_1 A(s_1) \right\| \leq  \lambda \Delta t, \quad \left\| \int_{t_j}^{t_j+\Delta t} d s_2 \int_{t_j}^{s_2} d s_1 A(s_2)A(s_1) \right\| \leq  \frac{(\lambda \Delta t)^2}{2}, \quad \cdots.
\end{align}
Thus, for sufficiently large $K$ (sufficiently small $\Delta t = t/K$), the distance of $\Xi_j$ from the identity is bounded, via triangle inequality, as
\begin{align}
    \| I - \Xi_j\| \leq \sum_{k \geq 1} \frac{(\lambda \Delta t)^k}{k!} = e^{\lambda \Delta t}-1 \approx \lambda \Delta t = \calO(1/K).
\end{align}
Thus, as described in the previous example, we can strategically block encode $\Xi_j$ such that $U_{\Xi_j}$ also has $\calO(1/K)$-bounded distance from the identity, thereby satisfying the $p$-MACG criterion.

\subsection{Oblivious Amplitude Amplification for Approximate Multiplication}
As discussed in the introduction, a key benefit of coherent multiplication of block encodings is that amplitude amplification can be used to boost the probability of the good measurement sequence. Recall that amplitude amplification (AA) is a Grover-style procedure, in which reflections are repeatedly performed around the desired ``winner'' state $\ket{\psi_\text{good}} = A_K \cdots A_2 A_1 \ket{\psi}$ and the intial state $\ket{\psi_0}$. In oblivious amplitude amplification (OAA) the ancillary state $\ket{0^{m+a}}$ signals the subspace of the desired block encoding, meaning reflections about the winner state are oblivious to the form of the state itself. For more detailed descriptions of AA and OAA, we refer the reader to \Cref{sec:aa} and \Cref{sec:oaa}, respectively. 

In \Cref{sec:oaa_embe}, we describe how OAA can be implemented with queries to an EMBE circuit, so as to boost the probability of the good measurement sequence to $\Omega(1)$. In \Cref{sec:oaa_ambe}, we further prove the following lemma, demonstrating how OAA can similarly boost the probability of the good measurement sequence for an $\eps$-AMBE circuit. 
\begin{lemma}[OAA for $\eps$-AMBE]
    Via repeated queries to an $\eps$-approximate multiplication of block encodings circuit, oblivious amplitude amplification can be used to boost the fidelity between the circuit output and the desired state to $1-O(\eps^2)$.
\end{lemma}

\section{Acknowledgements}
FV would like to acknowledge Louis Golowich, Nathan Ju, Lin Lin, Zeph Landau, Fermi Ma, Jelani Nelson, and Ewin Tang for helpful technical discussions. FV would, in particular, also like to thank Zeph Landau for helpful feedback on the manuscript. AG is supported by the EU’s Horizon 2020 Marie
Sklodowska-Curie program QuantOrder-89188. FV is supported by the Paul and Daisy Soros Fellowship for New Americans, NSF Graduate Research Fellowship Grant \#DGE2146752, and NSF Grant \#2311733. This work was done in part while the authors visited the Simons Institute for the Theory of Computing, supported by DOE QSA Grant \#FP00010905.

\newpage
\bibliographystyle{alphaUrlePrint.bst}
\bibliography{qc_gily}

\newpage
\appendix

\section{Part I: Approximate Single-Ancilla Block Encodings} \label{sec:part_i}
In this section, we will show how to uncompute all but one of a block encoding's ancillae. The uncomputed ancillae can then be reused for other tasks (beyond implementing the block encoding) in a quantum algorithm. This establishes a ``space-time tradeoff'' in which we can free up space (ancillae) via computation (time). Alternatively, this demonstrates that for any matrix $\general$, there exists an $\epsilon$-precise block-encoding of $\general$ requiring only one ancilla qubit.

In \Cref{sec:hermitian_uncompute} we describe this uncomputation procedure, in detail, for single-ancilla block encodings of Hermitian matrices $\herm$. In \Cref{sec:non_hermitian} we describe how this procedure can be generalized to single-ancilla block encodings of general matrices $\general$. 

\subsection{Single-Ancilla Block Encodings of Hermitian Matrices} \label{sec:hermitian_uncompute}

Let $V_{\herm}$ be a $(1, a, 0)$-block-encoding of $n$-qubit Hermitian matrix $\herm$, where $\norm{\herm} \leq 1-\delta$ for any $\delta \in (0,1)$.
In this section, we will show that we can efficiently map $V_{\herm}$, which requires several ancilla qubits, into $U_{\herm}$, an $\calO(\eps)$-block encoding of $\herm$ requiring only a single ancilla, where $\epsilon \in (0,1)$. 
In particular, we will prove the following main theorem.
\mainUncompute*
\begin{proof}
We will propose a specific quantum circuit (parameterized by $\delta$ and $\epsilon$) that implements an $(n+a+c)$-qubit matrix $\widetilde{W}_\herm$ (for $c=\calO(1))$, which is a $(1, a+c-1, \epsilon)$-block encoding of the $(n+1)$-qubit canonical block encoding matrix
\begin{align}
    U_{\herm} = Z \otimes \herm + X \otimes \sqrt{I - \herm^2}= \begin{pmatrix}
        \herm & \sqrt{I - \herm^2} \\
        \sqrt{I - \herm^2} & -\herm
    \end{pmatrix}.
\end{align}
We will now prove that $U_{\herm}$ is a $(1,1,0)$-block encoding of $\herm$. (To see this intuitively, consider the trivial case $\herm = \cos(x)$).
\begin{claim} \label{thm:hermitian}
Let $\herm$ be any Hermitian matrix with operator norm at most 1. Then, the matrix 
\[
U_{\herm} = Z \otimes \herm + X \otimes \sqrt{I - \herm^2}
\]
is a single-qubit Hermitian unitary block-encoding of $ \herm $, in the sense that $\herm= (\bra{0}\otimes I) U_\herm (\ket{0}\otimes I)$.
\end{claim}
\begin{proof}[Proof of \Cref{thm:hermitian}]
Note that $ Z $ and $ X $ are Pauli matrices, which are Hermitian, and $\sqrt{I - \herm^2}$ is also Hermitian. 
Since the tensor product and sum of Hermitian matrices is Hermitian, $U$ is Hermitian.

Next, consider the block structure of $ U $:
\begin{align}\label{eq:HermitianU}
    U_{\herm} = \begin{pmatrix}
        \herm & \sqrt{I - \herm^2} \\
        \sqrt{I - \herm^2} & -\herm
    \end{pmatrix}.
\end{align}
This matrix has $ \herm $ as the block component we are interested in and is thus a block-encoding of $ \herm $. Now, we verify that $ U_{\herm} $ is unitary by showing that
\begin{align}
U_{\herm}^\dagger U_{\herm} = U_{\herm}^2 = (Z \otimes \herm + X \otimes \sqrt{I - \herm^2})^2 = I.
\end{align}
In the last equality we used the properties of Pauli matrices, $ Z^2 = X^2 = I $, and $ ZX = -XZ $~to~get:
\begin{align}
U_{\herm}^2 = Z^2 \otimes \herm^2 + X^2 \otimes (I - \herm^2) + (ZX + XZ) \otimes (\herm \sqrt{I - \herm^2}) = I \otimes \herm^2 + I \otimes (I - \herm^2) = I. 
\end{align}
This concludes the proof of \Cref{thm:hermitian}.
\end{proof}
\noindent With this, we will now prove that the matrix block encoded by $\widetilde{W}_{\herm}$, which we will denote $\widetilde{U}_{\herm}$, is itself a $(1,1,\epsilon)$-block encoding of $\herm$.
\begin{claim} \label{thm:be}
    The matrix $\widetilde{U}_{\herm}=\bra{\bar{0}^{a+c-1}} \widetilde{W}_\herm\ket{\bar{0}^{a+c-1}}$ is a $(1,1,\epsilon)$-block encoding of $\herm$.
\end{claim}
\begin{proof}[Proof of \Cref{thm:be}]
    Let $\widetilde{U}_{\herm}$ denote the $n+1$ qubit matrix block-encoded by $\widetilde{W}_\herm$, i.e.
    \begin{align}
        \widetilde{U}_{\herm} = \left(\bra{0^{a+c-1}} \otimes I^{n+1}\right) \widetilde{W}_\herm \left(\ket{0^{a+c-1}} \otimes I^{n+1}\right).
    \end{align}
    Since $\widetilde{W}_\herm$ is an $\epsilon$-precise block encoding of $U_\herm$, i.e. $\| U_\herm - \widetilde{U}_{\herm} \| \leq \epsilon$, this implies that, i.e.
    \begin{align}
         \| \herm - \bra{0}\widetilde{U}_{\herm}\ket{0} \| =\| \bra{0}U_\herm\ket{0} - \bra{0}\widetilde{U}_{\herm}\ket{0} \|\leq \| U_\herm - \widetilde{U}_{\herm} \| \leq \eps.
    \end{align}
    Thus, $\widetilde{U}_{\herm}$ is a $(1,1,\eps)$-block encoding of $\herm$.
\end{proof}
\noindent Therefore, by computing $\widetilde{W}_\herm$ and post-selecting for the $a+c-1$-qubit measurement corresponding to block-encoding $\widetilde{U}_{\herm}$, we obtain a $(1,1,\eps)$-block encoding of $\herm$, as desired. 

The remainder of the proof will be dedicated to specifying a quantum algorithm
that computes $\widetilde{W}_\herm$. Specifically, we will formally prove each of the steps delineated by the high-level pseudocode of \Cref{alg:uncompute}.

\paragraph{Step 1.} We begin by leveraging linear combination of unitaries (LCU) to construct an explicit circuit mapping $V_\herm$ to $V_{\frac{I-\herm^2}{2}}$, which is a $(1,a+1,0)$-block encoding of the matrix $\frac{I-\herm^2}{2}$. This circuit requires two controlled-$V_\herm$ queries.
\begin{claim} \label{thm:lcu_1}
    The circuit 
    \begin{align} \label{eqn:i-h_be}
        V_{\frac{I-\herm^2}{2}} = (XH \otimes I_a)~C_{\Pi_1} V_{\herm}^2~(H \otimes I_a)
    \end{align}
    is a $(1,a+1,0)$-block-encoding of $\frac{I-\herm^2}{2}$.
\end{claim}
\begin{proof}[Proof of \Cref{thm:lcu_1}]
    Since $V_{\herm}$ is a $(1,a,0)$-block-encoding of $\herm$, by \cite[Lemma 54]{gilyen2018QSingValTransf}, $V_{\herm}^2$ is a $(1,a,0)$-block-encoding of $\herm^2$. Furthermore, the circuit defined in \Cref{eqn:i-h_be} is a $(1,1,0)$-block-encoding of $\frac{I-V_{\herm}^2}{2}$. Note that the unitary $\frac{I-V_{\herm}^2}{2}$ is a $(1,1,0)$-block encoding of $\frac{I-\herm^2}{2}$.
\end{proof}

\paragraph{Step 2.} With this, we can now leverage quantum singular value transformation (QSVT) to approximately apply a function to the singular values of $V_{\frac{I-\herm^2}{2}}$, implementing the matrix $V_{\frac{\sqrt{I - \herm^2}}{\sqrt{8}}}$, which is a $(1,a+\calO(1),\frac{\eps}{9})$-block-encoding of  $\frac{\sqrt{I-\herm^2}}{\sqrt{8}}$. 
\begin{claim} \label{thm:qsvt}
    The QSVT can be used to approximately apply the function $f(x) := \frac12 \sqrt{x}$ to $V_{\frac{I-\herm^2}{2}}$ so as to implement the unitary $V_{\frac{\sqrt{I - \herm^2}}{\sqrt{8}}}$, which is a $(1,a+\calO(1),\frac{\eps}{9})$-block-encoding of  $\frac{\sqrt{I-\herm^2}}{\sqrt{8}}$.  
\end{claim}
\begin{proof}[Proof of \Cref{thm:qsvt}]
    Since $\|\herm\| \leq 1-\delta$, we have that $I-\herm^2 \succeq 1-\delta$. Therefore, when applying functions to the singular values of $\frac{I-\herm^2}{2}$, we only need to consider the action of the function on the domain $[\delta, 1]$.

    The QSVT can be used to apply many polynomials to block-encoded matrices, with arbitrary precision. Specifically, in the case of functions of degree $c$ such that  $0< c\leq 1$, we can leverage the following result.
    \begin{fact}[\cite{gilyen2018QSingValTransfThesis} Corollary 3.4.14] \label{thm:qsvt_err}
    Let $\delta,\eta\in (0,\frac12]$, $c \in (0,1]$, and $f(x) := \frac12 x^c$. Then, there exist even/odd polynomials $P,P' \in \mathbb{R}[x]$ such that $\norm{P-f}_{[\delta,1]}\leq \eta$, $\norm{P}_{[-1,1]}\leq 1$ and $\norm{P'-f}_{[\delta,1]}\leq \eta$, $\norm{P'}_{[-1,1]}\leq 1$ \footnote{Note that for a polynomial $Q$, $\norm{Q}_{[a,b]}$ denotes the norm of $Q(x)$ in the interval $x\in [a,b]$.}. Moreover, the degree of the polynomials is $\bigO{\frac{1}{\delta}\log\left(\frac{1}{\eta}\right)}$.
    \end{fact}

    By \Cref{thm:qsvt_err}, there exists a degree-$\bigO{\frac{1}{\delta}\log\left(\frac{1}{\eps}\right)}$ polynomial $P^*$, such that $\norm{P^*}_{[-1,1]}\leq 1$ and, for $f(x) := \frac12 \sqrt{x}$,
    \begin{align} \label{eqn:poly_err}
        \norm{P^*-f}_{[\delta,1]}\leq \frac{\eps}{9}.
    \end{align}
    Via the QSVT, we can use $\calO(1)$ additional ancillae to apply the polynomial $P^*$ to the singular values of $V_{\frac{I-\herm^2}{2}}$. We will denote the output matrix of this QSVT operation as $V_{\frac{\sqrt{I - \herm^2}}{\sqrt{8}}} = P^* \left(V_{\frac{I-\herm^2}{2}}\right)$. Leveraging the triangle inequality, the error induced by this QSVT procedure is bounded as
    \begin{align}
        &\norm{\bra{0^{a+2}} V_{\frac{\sqrt{I - \herm^2}}{\sqrt{8}}} \ket{0^{a+2}} - \frac{\sqrt{I - \herm^2}}{\sqrt{8}}} = \norm{\bra{0^{a+2}} P^* \left(V_{\frac{I-\herm^2}{2}}\right) \ket{0^{a+2}} - \frac{\sqrt{I - \herm^2}}{\sqrt{8}}} \\
        & \leq \norm{ \bra{0^{a+2}} P^* \left(V_{\frac{I-\herm^2}{2}}\right) - f \left(V_{\frac{I-\herm^2}{2}}\right) \ket{0^{a+2}} } + \norm{\bra{0^{a+2}} f \left(V_{\frac{I-\herm^2}{2}}\right) \ket{0^{a+2}} - \frac{\sqrt{I - \herm^2}}{\sqrt{8}} } \\
        & \leq \norm{ P^* \left(V_{\frac{I-\herm^2}{2}}\right) - f \left(V_{\frac{I-\herm^2}{2}}\right) } + \norm{f \left(\bra{0^{a+2}} V_{\frac{I-\herm^2}{2}}\ket{0^{a+2}}\right)  - \frac{\sqrt{I - \herm^2}}{\sqrt{8}} } \\
        & \leq \frac{\eps}{9} + \norm{  \frac{\sqrt{I - \herm^2}}{\sqrt{8}}  - \frac{\sqrt{I - \herm^2}}{\sqrt{8}} } =  \frac{\eps}{9}. \label{eqn:temp} 
    \end{align}
    This concludes the proof of \Cref{thm:qsvt}.
\end{proof}

\paragraph{Step 3.} With a subroutine to compute $V_{\frac{\sqrt{I - \herm^2}}{\sqrt{8}}}$, we will leverage linear combination of unitaries again in order to compute a 
matrix $W_{U_\herm}$, which is a $(1,a+\calO(1),\frac{\eps}{14})$-block-encoding of $\sin(\pi/14) \cdot U_\herm$.
\begin{lemma} \label{thm:lcu}
    Linear combination of unitaries with $V_\herm$ and $V_{\frac{\sqrt{I - \herm^2}}{\sqrt{8}}}$, can be used to implement
    \begin{align}
        W_{U_\herm} =\sqrt{8}\sin(\pi/14) \cdot X \otimes V_{\frac{\sqrt{I - \herm^2}}{\sqrt{8}}}  + \sin(\pi/14) \cdot Z  \otimes V_{\herm},
    \end{align}
    which is a $(1,a+\calO(1),\frac{\eps}{14})$-block-encoding of $\sin(\pi/14) \cdot U_\herm$.
\end{lemma}
\begin{proof}[Proof of \Cref{thm:lcu}]
    Recall that $V_{\frac{\sqrt{I - \herm^2}}{\sqrt{8}}}$ is a $(1,a+\calO(1),\frac{\eps}{9})$-block-encoding of  $\frac{\sqrt{I-\herm^2}}{\sqrt{8}}$, operating on $n+a+\calO(1)$ qubits. Meanwhile, $V_{\herm}$ is a $(1, a, 0)$-block-encoding of $\herm$, operating on $n+a$ qubits. 

    Performing the standard LCU procedure \cite{childs2012hamiltonian} with the selection oracle
    \begin{align}
        U_\text{SEL} = \ketbra{0}{0} \otimes  X \otimes \widetilde{V}_{\sqrt{I-\herm^2}} + \ketbra{1}{1} \otimes Z \otimes I^{\otimes \calO(1)} \otimes V_{\herm}
    \end{align}
    and preparation oracle
    \begin{align}
        V_\text{PREP} \ket{0} = \frac{1}{3 \sqrt{\sin(\pi/14)}} \left( \sqrt{8 \sin(\pi/14)} \ket{0} + \sqrt{ \sin(\pi/14)} \ket{1}\right),
    \end{align}
    we can implement
    \begin{align}
        W_{U_\herm} &=V_\text{PREP}~U_\text{SEL}~V_\text{PREP}^\dagger = \sin(\pi/14) \cdot \left(\sqrt{8} \cdot X \otimes \widetilde{V}_{\sqrt{I-\herm^2}}  +  Z \otimes I^{\otimes \calO(1)}\otimes V_{\herm} \right) 
    \end{align}
    According to the selection $\ket{\bar{0}}=\ket{0^{a+\calO(1)}}$, this matrix block-encodes
    \begin{align}
        \bra{\bar{0}}W_{U_\herm} \ket{\bar{0}} &= \sin(\pi/14) \cdot \left(\sqrt{8} \cdot X \otimes \bra{\bar{0}}V_{\frac{\sqrt{I - \herm^2}}{\sqrt{8}}} \ket{\bar{0}}  +  Z \otimes \bra{0^a}V_{\herm} \ket{0^a} \right). 
    \end{align}
    Using the facts that $\norm{A \otimes B} \leq \norm{A}\norm{B}$, $\norm{U}=1$ for any unitary matrix $U$, and $\norm{a \cdot A}=|a|\cdot \norm{A}$ for any scalar $a$, we will now prove that  $\bra{\bar{0}}W_{U_\herm} \ket{\bar{0}}$ is $\frac{\eps}{14}$-close to $U_\herm$.
    \begin{align}
        &\norm{\sin(\pi/14) \cdot U_\herm - \bra{\bar{0}}W_{U_\herm} \ket{\bar{0}}} \\
        &= \sin(\pi/14) \cdot \norm{(X \otimes \sqrt{I - \herm^2}+Z \otimes \herm) - \left(\sqrt{8} \cdot X \otimes \bra{\bar{0}}V_{\frac{\sqrt{I - \herm^2}}{\sqrt{8}}} \ket{\bar{0}}  +  Z \otimes \bra{0^a}V_{\herm} \ket{0^a} \right)} \\
        &= \sin(\pi/14) \cdot \norm{X \otimes \left(\sqrt{I - \herm^2} - \sqrt{8} \cdot \bra{\bar{0}}V_{\frac{\sqrt{I - \herm^2}}{\sqrt{8}}} \ket{\bar{0}}\right)+Z \otimes \left(\herm -  \herm \right)} \\
        & \leq \sin(\pi/14) \cdot \left( \norm{X \otimes \left(\sqrt{I - \herm^1} - \sqrt{8} \cdot \bra{\bar{0}}V_{\frac{\sqrt{I - \herm^2}}{\sqrt{8}}} \ket{\bar{0}}\right)} \right) \\
        & \leq \sin(\pi/14) \cdot \left( \norm{X} \norm{\sqrt{I - \herm^2} - \sqrt{8} \cdot \bra{\bar{0}}V_{\frac{\sqrt{I - \herm^2}}{\sqrt{8}}} \ket{\bar{0}}} \right) \\
        &  = \sin(\pi/14) \cdot \left( \sqrt{8} \cdot \norm{\frac{\sqrt{I - \herm^2}}{\sqrt{8}} - \bra{\bar{0}}V_{\frac{\sqrt{I - \herm^2}}{\sqrt{8}}} \ket{\bar{0}}} \right) \\
        &  \leq \sin(\pi/14) \cdot \left( \sqrt{8} \cdot \frac{\eps}{9} \right) \\
        & \leq \frac{\eps}{14}
    \end{align}
    This concludes the proof of \Cref{thm:lcu}.
\end{proof}

\paragraph{Step 4.} In order to produce a block-encoding that is $\eps$-close to $U_\herm$, we will perform oblivious amplitude amplification (OAA) on $W_{U_\herm}$. The following result guarantees our ability to amplify the desired coefficient.

\begin{fact}[\cite{gilyen2018QSingValTransf} Theorem 28] \label{thm:oaa}
Let $n \in \mathbb{N}_+$ be odd, let $\eps\in\mathbb{R}_+$, let $Q$ be a unitary, let $\widetilde{\Pi}$, $\Pi$ be orthogonal projectors and let $M: \textnormal{img}(\Pi) \mapsto \textnormal{img}(\widetilde{\Pi})$ be an isometry, such that 
\begin{align}
    \norm{\sin\left(\frac{\pi}{2n}\right)M\ket{\psi}-\widetilde{\Pi} Q \ket{\psi}} \leq \eps
\end{align}
for all $\ket{\psi} \in \textnormal{img}(\Pi)$. Then, we can construct a unitary $\widetilde{Q}$ such that for all $\ket{\psi} \in \textnormal{img}(\Pi)$,
\begin{align}
    \norm{M \ket{\psi}-\widetilde{\Pi} \widetilde{Q} \ket{\psi}} \leq 2n \eps,
\end{align}
which uses a single ancilla qubit, with $n$ uses of $Q$ and $Q^\dagger$, $n$ uses of $C_\Pi NOT$ and $n$ uses of $C_{\widetilde{\Pi}} NOT$ gates and $n$ single qubit gates.
\end{fact}
\noindent Specifically, in our case, we have that
\begin{align}\label{eq:closeBlockEncoding}
    \norm{\sin\left(\frac{\pi}{14}\right)\cdot U_{\herm}\ket{\psi}- \bra{\bar{0}}W_{U_\herm} \ket{\bar{0}} \ket{\psi}} \leq \frac{\eps}{14}.
\end{align}
Thus, by \Cref{thm:oaa}, we can implement matrix $\widetilde{W}_{U_\herm}$, which is an $\eps$-precise block-encoding of $U_{\herm}$, i.e.
\begin{align}
    \norm{U_{\herm}\ket{\psi}- \bra{\bar{0}}\widetilde{W}_{U_\herm} \ket{\bar{0}} \ket{\psi}} \leq \eps.
\end{align}
In practice, $\widetilde{W}_{U_\herm}$ is implemented via the QSVT. Using one ancilla qubit,  four $W_{U_\herm}$ queries, and three $W_{U_\herm}^\dagger$ queries a Chebyshev polynomial of $7^\text{th}$ order is applied to the block-encoded matrix. Since $T_7 \left(\sin\left(\frac{\pi}{14}\right)\right)=-1$, this amplifies the amplitude of the matrix block-encoded by $W_{U_\herm}$ to the desired amplitude (up to a phase).

By \Cref{thm:oaa}, we know that the top left corner of the implemented matrix is close to the desired. It is not too difficult to see that this in particular implies that the implemented approximate block encoding is $\bigO{\eps}$-close to an ideal ancilla-saving unitary $U$ for which $(\bra{\bar{0}}\otimes I)U(\ket{\bar{0}}\otimes I)=U_{\herm}$, see for example~\cite[Lemma 23]{gilyen2018QSingValTransf} and its proof for a detailed argument.

Indeed, \Cref{eq:closeBlockEncoding} combined with the proof of~\cite[Lemma 23]{gilyen2018QSingValTransf} shows that there is a unitary $V$ such that $\bra{\bar{0}} V \ket{\bar{0}}=\sin\left(\frac{\pi}{14}\right)\cdot U_{\herm}$ and $\norm{W_{U_\herm} - V}=\bigO{\eps}$. If we were using $V$ in the oblivious amplitude amplification we would end up with a unitary whose top left corner is exactly $U_\herm$, and that unitary is $\bigO{n\eps}$ close to $\widetilde{W}_{U_\herm}$.
\end{proof}

\subsection{Single-Ancilla Block Encodings of General Matrices} \label{sec:non_hermitian}
In the previous section, we focused on single-ancilla block encodings of Hermitian matrices $\herm$. We will now show how an analogous procedure can be leveraged to implement a single-ancilla block-encoding of general sub-normalized matrices $\general$.
\generalCorollary*
\begin{proof}
    The critical difference in the procedure for general matrices $\general$ from that of Hermitian matrices $\herm$ (\Cref{thm:main_uncompute}) is in the form of the single-ancilla block encoding to be implemented. For general matrices we will implement a block encoding of the following Hermitian matrix, which we also prove to be unitary. (To get some intuition for the following statement it is worth considering the trivial case $\general=\sin(x)$.)
    \begin{claim} \label{thm:hermitian_general} 
        Let $\general$ be any matrix with operator norm at most 1. Then, the matrix 
        \begin{align}
            U_\general = \begin{pmatrix}
            \sqrt{I - \general^\dagger \general} & \general^\dagger \\
            \general & -\sqrt{I - \general \general^\dagger}
            \end{pmatrix}
        \end{align}
        is a Hermitian unitary block-encoding of $\general$ in the sense that $\general= (\bra{1}\otimes I) U_\general (\ket{0}\otimes I)$.
    \end{claim}
    \begin{proof}[Proof of \Cref{thm:hermitian_general}] 
        $U_\general$ is Hermitian by design, since $ \sqrt{I - \general^\dagger \general} $ and $ \sqrt{I - \general \general^\dagger} $ are Hermitian. 
        
        We will now prove that $ U $ is unitary. To do so, we leverage the following facts,
        \begin{align}
            \sqrt{I - \general^\dagger \general} \general^\dagger &= \general^\dagger \sqrt{I - \general \general^\dagger} \\
            \quad \general \sqrt{I - \general^\dagger \general} &= \sqrt{I - \general \general^\dagger} \general, \\
            \left(\sqrt{I - \general^\dagger \general}\right)^2 &+ \general^\dagger \general = I \\
            \left(\sqrt{I - \general \general^\dagger}\right)^2&+ \general \general^\dagger = I
        \end{align}
        which can be verified via the singular value decomposition of $\general$. Thus,
        \begin{align}U_\general^\dagger U_\general = U_\general^2 
        &= \begin{pmatrix}
        \left(\sqrt{I - \general^\dagger \general}\right)^2 + \general^\dagger \general & \sqrt{I - \general^\dagger \general} \general^\dagger - \general^\dagger \sqrt{I - \general \general^\dagger} \\
        \general \sqrt{I - \general^\dagger \general} - \sqrt{I - \general \general^\dagger} \general & \general \general^\dagger + \left(\sqrt{I - \general \general^\dagger}\right)^2
        \end{pmatrix} \\
        & = \begin{pmatrix}
        I & \general^\dagger \left( \sqrt{I - \general \general^\dagger}-\sqrt{I - \general \general^\dagger} \right)\\
        \general \left(\sqrt{I - \general^\dagger \general} -\sqrt{I - \general^\dagger \general}\right) & I
        \end{pmatrix}
         = I,
        \end{align}
        confirming that $U_\general$ is in fact unitary.
    \end{proof}
    The protocol for implementing $U_\general$ follows closely to that of $U_\herm$. However, we offer a brief outline for completeness. We are given access to $V_\general$, an $a$-ancilla block encoding of $\general$, as well as $V_\general^\dagger$, which block encodes $\general^\dagger$. Linear combination of unitaries is used to implement block encodings of $I-\general \general^\dagger$ and $I- \general^\dagger\general$. Application of the QSVT to these block encodings results in approximate block encodings of $\sqrt{I-\general \general^\dagger}$ and  $\sqrt{I- \general^\dagger\general}$. Via linear combination of unitaries, these are combined into an approximate block encoding of $\alpha \cdot U_\general$, for a coefficient $\alpha$. Since $U_\general$ is a unitary matrix, oblivious amplitude amplification is used to boost the amplitude $\alpha$, resulting in the desired approximate block encoding of $U_\general$. 
\end{proof}

\section{Part II: Approximate Multiplication of Block Encodings} \label{sec:part_ii}

In general, an $(n+a)$-qubit block encoding can be defined for any bit-strings $b_1, b_2 \in \{0,1\}^a$,
\begin{align}
    A = (\bra{b_1} \otimes I_n) U_A (\ket{b_2} \otimes I_n).
\end{align}
Note, however, that we can straightforwardly modify the block encoding unitary as
\begin{align}
    \widetilde{U}_A = (X^{b_1} \otimes I_n) U_A (X^{b_2} \otimes I_n)
\end{align}
to obtain a new unitary $\widetilde{U}_A$ such that
\begin{align}
    A = (\bra{0^a} \otimes I_n) \widetilde{U}_A (\ket{0^a} \otimes I_n).
\end{align}
Therefore, without loss of generality, in this section, we will only consider block encodings where $A$ is embedded in the top left corner of the unitary. Where obvious, we will also not explicitly write the $I_n$ terms to improve readability.

For a series of $K$ $(n+a)$-qubit block encodings $(U_{A_i})_{i=1}^K$, we will denote the desired multiplication of block encodings as
\begin{align}
    A_{[K]} = A_K \cdots A_2 A_1.
\end{align}
In terms of the full unitaries, this can be decomposed as  
\begin{align} \label{eqn:multibe}
    A_{[K]} &= \bra{0^a} U_{A_K} \ketbra{0^a}{0^a} \cdots \ketbra{0^a}{0^a} U_{A_2} \ketbra{0^a}{0^a} U_{A_1} \ket{0^a} \\
    & = \bra{0^a} U_{A_K} \pzeroa \cdots \pzeroa U_{A_2} \pzeroa U_{A_1} \ket{0^a},
\end{align}
where $\pzeroa = \ketbra{0^a}{0^a}$ with orthogonal projector $\pperp = I - \pzeroa$
However, the projectors $\pzeroa$ are not unitary and can thus only be implemented either by performing measurements incoherently or coherently block encoding them. 

A coherent block encoding of the projector $\pzeroa$ takes the general form
\begin{align}
    U_{\pzeroa} = G \otimes \pzeroa + B \otimes \pperp,
\end{align}
where  $G$ and $B$ are both unitaries. We will prove that this construction is in fact unitary. 
\begin{claim} \label{thm:unitary}
    For arbitrary unitaries $G$ and $B$, $U_{\pzeroa} = G \otimes \pzeroa + B \otimes \pperp$ is unitary.
\end{claim}
\begin{proof}[Proof of \Cref{thm:unitary}]
    To confirm that $U_{\pzeroa}$ is unitary, we must show that $U_{\pzeroa}^\dagger U_{\pzeroa} = U_{\pzeroa} U_{\pzeroa}^\dagger = I$.
    \begin{align}
        U_{\pzeroa}^\dagger U_{\pzeroa} &= G^\dagger G \otimes \pzeroa + B^\dagger B \otimes \pperp = I \otimes (\pzeroa + \pperp) = I \\
        U_{\pzeroa}^\dagger U_{\pzeroa} &= G G^\dagger \otimes \pzeroa + B B^\dagger \otimes \pperp = I \otimes (\pzeroa + \pperp)=I
    \end{align}
    This concludes the proof of \Cref{thm:unitary}.
\end{proof}

If $G$ and $B$ are $m$-qubit unitaries, $U_{\pzeroa}$ specifically block-encodes the matrix
\begin{align}
    \bra{0^m} U_{\pzeroa} \ket{0^m} = \bra{0^m} G \ket{0^m} \cdot \pzeroa + \bra{0^m} B \ket{0^m} \cdot \pperp.
\end{align}
Thus, by setting $m=1$, $G = I$ and $B=X$, we obtain the desired block-encoding $\bra{0} U_{\pzeroa} \ket{0} = \pzeroa$
of a single $\pzeroa$ projector.

Returning to \Cref{eqn:multibe}, to implement the multiplication of $K$ block encodings, an $m$-qubit block encoding of $A_{[K]}$ takes the general form 
\begin{align} \label{eqn:encode_circuit}
    U_{A_{[K]}} = U_{A_K} (G_{K-1} \otimes \pzeroa + B_{K-1} \otimes \pperp) U_{A_{K-1}} \cdots  U_{A_2} (G_1 \otimes \pzeroa + B_1 \otimes \pperp) U_{A_1}.
\end{align}
Specifically, this matrix block encodes
\begin{align}
    \bra{0^m} U_{A_{[K]}} \ket{0^m} &= \sum_{x \in \{0,1\}^{K-1}} \bra{0^m} \prod_{i=1}^{K-1} G_i^{x_i} B_i^{1-x_i} \ket{0^m} \cdot U_{A_K}\prod_{i=1}^{K-1} \pzeroa^{x_i} \pperp^{1-x_i} U_{A_i} \\
    & = \bra{0^m} G_{K-1}\cdots G_1 \ket{0^m} \cdot U_{A_K} \pzeroa U_{A_{K-1}} \cdots  U_{A_2} \pzeroa U_{A_1} \\
    & \quad \quad \quad + \sum_{x \in \{0,1\}^{K-1}\backslash 0^{K-1}} \bra{0^m} \prod_{i=1}^{K-1} G_i^{x_i} B_i^{1-x_i} \ket{0^m} \cdot U_{A_K}\prod_{i=1}^{K-1} \pzeroa^{x_i} \pperp^{1-x_i} U_{A_i} \\
    & = \bra{0^m} \prod_{i=1}^{K-1} G_i \ket{0^m} \cdot A_{[K]} \\
    & \quad\quad\quad+ \sum_{x \in \{0,1\}^{K-1}\backslash 0^{K-1}} \bra{0^m} \prod_{i=1}^{K-1} G_i^{x_i} B_i^{1-x_i} \ket{0^m} \cdot U_{A_K}\prod_{i=1}^{K-1} \pzeroa^{x_i} \pperp^{1-x_i} U_{A_i}
\end{align}
For large enough $m$ and careful choice of the unitaries $(G_i)_{i=1}^{K-1}$ and $(B_i)_{i=1}^{K-1}$, this construction is guaranteed to implement the desired multiplication of block encodings, such that 
\begin{align} 
    \bra{0^{m+a}} U_{A_{[K]}} \ket{0^{m+a}} =\bra{0^a}U_{A_K} \pzeroa \cdots  \pzeroa U_{A_2} \pzeroa U_{A_1}\ket{0^a}= A_{[K]}.
\end{align}
The key question that we explore in this section is how small we can make the number of measurement ancillae $m$, while still achieving this desired block encoding (or something close to it).

\subsection{The Multiple Coherent Measurement (\circclass)~Circuit Class}

We begin by formalizing the Multiple Coherent Measurement  (\circclass) circuit class, which implements unitaries of the form of $U_{A_{[K]}}$, as described in \Cref{eqn:encode_circuit}. This circuit class will serve as the central object of study in Part II. We will show that all known procedures for implementing multiplication of block encodings can be implemented in this circuit class. Furthermore, restriction to this circuit class will enable us to prove a lower-bound on the number of measurement ancillae $m$ required to achieve \emph{exact} multiplication of block encodings (\Cref{def:exact_mbe}).
\begin{definition}[\circclass$_{\vec{U}_A, K, \vec{\Pi}}$ Circuit Class] 
    Suppose we are given series of $K$ $(n+a)$-qubit block-encodings $\vec{U}_A = (U_{A_i})_{i=1}^K$. Without loss of generality, assume that all $U_{A_i}$ have $\pzeroa$ as the $a$-qubit projector onto ``good'' measurement outcome and $\pperp=I_a-\pzeroa$ as the $a$-qubit projector onto ``bad'' measurement outcome. With this, we define $\circclass_{\vec{U}_A, K,m}$ to be the class of all circuits of the form,
    \begin{align} \label{eqn:mcm_general}
        \circclass_{\vec{U}_A, K,m} (\vec{W}, \vec{G}, \vec{B}) = (W_{K-1} \otimes U_K) \prod_{i=1}^{K-1}   \left( \left(G_{K-i} \otimes \pzeroa \otimes I_n \right)\left( B_{K-i} \otimes \pperp \otimes I_n \right)    \left(W_{K-i-1} \otimes U_{K-i}\right) \right),
    \end{align}
    parameterized by the series of $m$-qubit unitaries $\vec{W} = (W_i)_{i=0}^{K-1}$, $\vec{G} = (G_i)_{i=1}^{K-1}$, and $\vec{B} = (B_i)_{i=1}^{K-1}$. The circuit is illustrated in \Cref{fig:mcm}.
\end{definition}
\begin{figure}[t!]
    \centering
    \includegraphics[width=.8\textwidth]{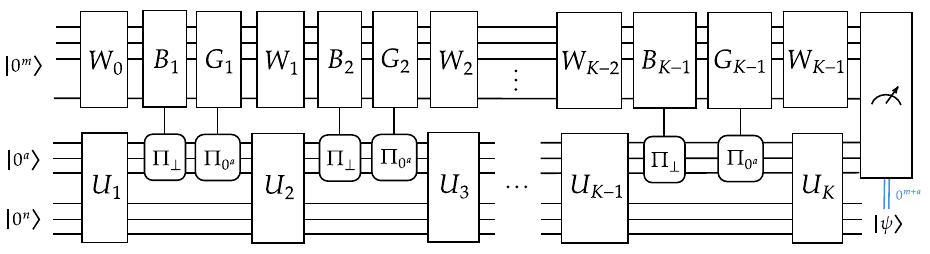}
    \caption{An $\circclass_{\vec{U}, K,m} (\vec{W}, \vec{G}, \vec{B})$ circuit.}
    \label{fig:mcm}
\end{figure}
\noindent  The \circclass$_{\vec{U}_A, K, m}$ circuit class was to defined to be as general possible, making it clear that it contains all circuits of the form of \Cref{eqn:encode_circuit}. However, as defined, the circuit class is overparameterized and its circuits can be simplified as follows.
\begin{lemma}[Simplified \circclass~Structure] \label{thm:simplify}
    For every $\circclass_{\vec{U}_A, K,m} (\vec{W}, \vec{G}, \vec{B})$ circuit, there exists a series of $m$-qubit unitaries $\vec{V} = (V_i)_{i=1}^{K-1}$ and  a single $m$-qubit unitary $Q$ such that $\circclass_{\vec{U}_A, K,m} (\vec{W}, \vec{G}, \vec{B})$ can be simplified to a circuit of the form
    \begin{align}
        \circclass_{\vec{U}_A, K,m} (\vec{V}, Q) = (Q\otimes U_K) \cdot \prod_{i=1}^{K-1} \left(  \left(V_{K-i}\otimes \pperp \otimes I_n \right) \left(I_{m+a} \otimes U_{K-i}\right)\right),
    \end{align}
    as depicted in \Cref{fig:lemma_con_uni}.
\end{lemma}
\begin{proof}[Proof of \Cref{thm:simplify}]
    For ease of notation, we will denote the unitary implemented by $\circclass_{\vec{U}_A, K,m} (\vec{W}, \vec{G}, \vec{B})$, in \Cref{eqn:mcm_general}, as $U_\circclass$. Leveraging the following identity
    \begin{align}
        (G_{j}\otimes \pzeroa) (B_{j}\otimes \pperp) &= \left( G_j \otimes \pzeroa + I \otimes\pperp \right) \left(I \otimes \pgood^{(j)} + B_j\otimes \pperp\right) \\
        &= G_j \left( I \otimes \pzeroa + G_j^\dagger \otimes\pperp \right) \left(I \otimes \pzeroa + B_j\otimes \pperp\right) \\
        &=  G_j \left(I  \otimes \pzeroa +  G_j^\dagger B_j \otimes \pperp\right). 
    \end{align}
    and setting $W_j^* = W_j G_{j+1}$ and $V_j^* = G_j^\dagger B_j$, the unitary simplifies to
    \begin{align}
        U_\circclass = (W_{K} W^*_{K-1} \otimes U_K) \prod_{i=1}^{K-1}  \left( \left( V^*_{K-i} \otimes \pperp \otimes I_n \right) \left(W_{K-i-1}^* \otimes U_{K-i}\right) \right),
    \end{align}
    which can be depicted as:
    \begin{figure}[H]
        \centering
        \includegraphics[width=.75\textwidth]{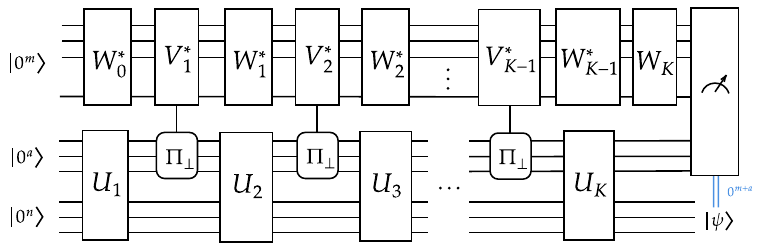}
    \end{figure}
    
    \noindent We will now recursively commute all the $W_i^*$ gates through all the $V_i^*$ gates to their right. Begin by defining the set of unitaries $V_{j+1}$ such that
    \begin{align}
        \left(\prod_{i=0}^j W^*_{j-1}\right) V_{j+1}  =  V^*_{j+1} \left(\prod_{i=0}^j W^*_{j-i}\right),
    \end{align}
    which implies that 
    \begin{align}
        (V^*_{j+1}\otimes \pperp) \left(\prod_{i=0}^j W^*_{j-i} \otimes I_a\right) &= \left( I  \otimes \pzeroa +  V^*_{j+1} \otimes \pperp  \right) \left(\prod_{i=0}^j W^*_{j-i} \otimes I_a\right)  \\
        &= \left(\prod_{i=0}^j W^*_{j-i}\right)  \otimes \pzeroa +  V^*_{j+1} \left(\prod_{i=0}^j W^*_{j-i}\right) \otimes \pperp \\
        &= \left(\prod_{i=0}^j W^*_{j-i}\right)  \otimes \pzeroa +   \left(\prod_{i=0}^j W^*_{j-i}\right) V_{j+1} \otimes \pperp \\
        &= \left(\prod_{i=0}^j W^*_{j-i} \otimes I_a\right) \left( I  \otimes \pzeroa +  V_{j+1} \otimes \pperp  \right) \\
        &=  \left(\prod_{i=0}^j W^*_{j-i} \otimes I_a\right) (V_{j+1}\otimes \pperp).
    \end{align}
    Thus, the circuit reduces to
    \begin{align}
        U_{\circclass} = (W_{K} \otimes I_{a+n}) \prod_{j=1}^{K} (W^*_{K-j} \otimes I_{a+n})  \prod_{i=0}^{K-1}  \left( \left( V_{K-i}\otimes \pperp \otimes I_n \right)  \left(I_m \otimes U_{K-i}\right) \right),
    \end{align}
    which is depicted as:
    \begin{figure}[H]
        \centering
        \includegraphics[width=.8\textwidth]{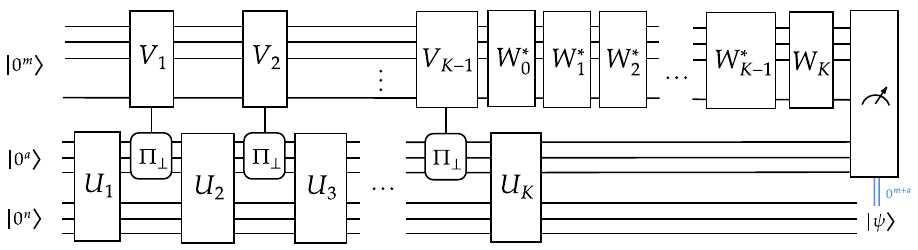}
    \end{figure}
    
    \noindent Therefore, setting $Q = W_K \prod_{j=1}^{K} W^*_{K-j} $, it is clear that $U_{\circclass} = \circclass_{\vec{U}_A, K,m} (\vec{V}, Q)$.
\end{proof}
\noindent From hereon, when we refer to the \circclass~circuit class, we will be referring to circuits of the form $\circclass_{\vec{U}_A, K,m} (\vec{V}, Q)$.

\subsection{Exact Compression Gadgets}
We will now introduce the notions of \emph{exact} multiplication of block encodings and \emph{exact} compression gadgets. Although multiplication of block encodings and compression gadgets are typically defined this way, we emphasize that they are exact to contrast with our eventual introduction of \emph{approximate} multiplication of block encodings and \emph{approximate} compression gadgets.

\begin{definition}[Exact Multiplication of Block Encodings (EMBE)] \label{def:exact_mbe}
    We will refer to an $\circclass_{\vec{U}_A, K,m} (\vec{V}, Q)$ circuit as an  \textbf{exact multiplication of block encodings (EMBE)} if it satisfies
    \begin{align} 
        \bra{0^{m+a}} \circclass_{\vec{U}_A, K,m} (\vec{V}, Q) \ket{0^{m+a}} = A_{[K]},
    \end{align}
    thereby exactly implementing the desired multiplication of block encodings.
\end{definition}

As described in the introduction (\Cref{sec:intro_part2}), the naive approach to implementing an exact multiplication of block encodings involves adding a measurement ancilla qubit for each block encoding. We can formally define this procedure in the \circclass~framework as follows.
\begin{definition}[Naive EMBE] \label{def:naive_dmbe}
    For a sequence of $K$ block encoding unitaries $(U_{A_i})_{i=1}^K$, a naive deterministic multiplication of block encodings can be implemented by a $\circclass_{\vec{U}_A, K,m} (\vec{V}, Q)$~circuit setting with $m=K-1$ measurement ancillae, $Q=I_m$, and   
    \begin{align}
        V_i  = X_i \otimes I_{[K-1]\backslash i} \quad \forall i \in [K-1].
    \end{align}
\end{definition}
\noindent However, as discussed in the introduction, this Naive EMBE is far from optimal in terms of the number of measurement ancillae $m$. This leads us to define the notion of an exact compression gadget\footnote{Note that the name ``compression gadget'' comes from the work of \cite{low2019Hamiltonian}.}. 
\begin{definition}[Exact Compression Gadget (ECG)] 
    For a sequence of $K$ block encoding unitaries $(U_{A_i})_{i=1}^K$,  any $\circclass_{\vec{U}_A, K,m} (\vec{V}, Q)$~circuit which implements an exact multiplication of block encodings using $m<K-1$ measurement ancillae is an \textbf{exact compression gadget (ECG)}.
\end{definition}

In this subsection, we will discuss the original ECG proposal of \cite{low2019Hamiltonian} and then prove a lower bound on the number of measurement ancillae required to implement \emph{any} ECG within the \circclass~circuit class. With this lower bound, we establish that the \cite{low2019Hamiltonian} ECG is in fact optimal in terms of number of measurement ancillae used. Therefore, the only way to use less measurement ancillae is to move to an approximate model of compression, as will be discussed in the next subsection.

\subsubsection{The \texorpdfstring{\cite{low2019Hamiltonian}}{[LW19]} Exact Compression Gadget} \label{sec:lw_dcg}

In their work on quantum simulation via truncated Dyson series, Low and Wiebe \cite{low2019Hamiltonian} proposed an ECG requiring only $m=\lceil \log_2 (K+1)\rceil$ measurement ancillae. Note, however, that their construction did not take advantage of the fact that we can determine whether the $k^{th}$ block-encoding was successful simply via the measurement outcome on the $a$-qubit block encoding ancilla register (it will have outcome $0^a$ if the block encoding was successful and outcome $x \in\{0,1\}^a\backslash 0^a$ if unsuccessful). Thus, we can implement their proposal within the \circclass~circuit class using only $\lceil \log_2 K \rceil$ measurement ancillae as follows (the gadget is a special case of the circuit depicted in \Cref{fig:add_gadget}).

\begin{definition}[\cite{low2019Hamiltonian} Exact Compression Gadget] \label{def:lw_dcg}
    For a sequence of $K$ block encoding unitaries $(U_{A_i})_{i=1}^K$, the \cite{low2019Hamiltonian} exact compression gadget is implemented by a $\circclass_{\vec{U}_A, K,m} (\vec{V}, Q)$~circuit with $m=\lceil \log_2 (K)\rceil$ measurement ancillae, $Q=I_m$, and   
    \begin{align}
        V_i  = \textnormal{\addm} \quad \forall i \in [K-1],
    \end{align}
    where  $\textnormal{\addm} \ket{c} = \ket{c+1 \bmod{2^m}}$ is the addition operation mod-$2^m$.
\end{definition}
To see gain intuition for why this DCG successfully implements a DMBE, note that the $m$-qubit measurment ancillae register can count bit strings up to $2^m = 2^{\lceil \log_2 (K+1)\rceil}\geq K+1$. The register is initialized to $\ket{0^m}$ and, for each bad ($\perp$) measurement, $+1$ is added, mapping to a basis state $\ket{x} \neq \ket{0}$. Furthermore, since at most, the \addm operation is applied $K$ times and the addition is at least mod-$(K+1)$, there will never be a case in which the modulo maps the sum back to the $\ket{0}$ state. Thus, if all measurements are good, the measurement register will remain in the $\ket{0^m}$ state and, for any bad measurement sequence, the measurement register will be in a basis state $\ket{x}\neq\ket{0}$. Therefore, come measurement time, if we measure the $\ket{0^m}$ state, the final state in the computation register is accepted. Otherwise, the final state is rejected.

\subsubsection{An Ancilla Lower-Bound for Exact Compression Gadgets} \label{sec:lb}
We will now prove a tight lower bound on the number of measurement ancillae necessary to implement an ECG and, thus, EMBE within the \circclass~circuit class.

\begin{theorem}(ECG Ancilla Lower-Bound) \label{thm:ancilla_lb}
    Any $\circclass_{\vec{U}_A, K,m} (\vec{V}, Q)$~circuit that implements an exact multiplication of $K$ block encodings requires  at least $m=\lceil\log_2(K)\rceil$ measurement ancillae.
\end{theorem}
\begin{remark}
    This lower bound is tight, since the modified \cite{low2019Hamiltonian} ECG saturates this bound.
\end{remark}
\begin{proof}[Proof of \Cref{thm:ancilla_lb}]
    In order to implement a DMBE, we must ensure that we have enough degrees of freedom (measurement ancillae) to distinguish two cases: 1) all intermediate measurements are $0^a$ and 2) any combination of intermediate measurements such that at least one is not $0^a$. 
    
    Note that we can directly determine whether or not the $K^{th}$ block-encoding was successful via our ultimate measurement of the $a$-qubit block encoding ancilla register. Therefore, we only need to determine the success of the first $K-1$ intermediate block encodings. If we think of each of these intermediate measurements as taking value $0^a$ or value $\perp=($not $0^a)$, then the set of all possible intermediate measurement outcomes is $\calM = \{0^a, \perp\}^{K-1}$, such that $|\calM|=2^{K-1}$. Our goal is to distinguish the $\{0^a\}^{K-1}$ outcome from the $\calM \backslash \{0^a\}^{K-1}$ outcomes. 
    
    Leveraging the \circclass~circuit structure, described in \Cref{thm:simplify}, we can reason about the number of ancillary degrees of freedom needed to distinguish these measurement outcomes. Let $\calS_0$ denote the intermediate measurement indices corresponding to good $0^a$ outcomes after the deferred measurement. We will define $\calS_\perp=[K]\backslash \calS_0$ to be indices corresponding to bad $\perp$ outcomes. Thus, before the deferred measurement, the measurement ancilla register of a $\circclass_{\vec{U}_A, \vec{\Pi}} (\vec{V}, Q)$~circuit will contain the state 
    \begin{align}
        \ket{\phi} = Q \left(\prod_{i \in \calS_\perp} V_i\right) \ket{0^m}.
    \end{align}
    Note that if all the measurements are good, i.e. $\calS_0=[K-1]$ and $\calS_\perp=\emptyset$, the measurement ancilla register will simply contain the state $\ket{\phi_0}=Q\ket{0^m}$. At the end of the computation, to distinguish ``good'' from ``bad'' measurement sequences, we must perform a projective measurement onto the basis $\{\good,\bad\}$, with projectors
    \begin{align}
        \good &:= Q\ketbra{0^m}{0^m} Q^\dagger  \\
        \bad &:= I- \good.
    \end{align}
    
    With this projective measurement basis, if all the measurements are good, then
    \begin{align}
        \bra{\phi_0}\good\ket{\phi_0} = \bra{0^m} Q^\dagger Q\ketbra{0^m}{0^m} Q^\dagger Q \ket{0^m}=1.
    \end{align}
    Therefore, in order to determinstically distinguish this state from all the bad measurement outcome sequences, there must exist sufficient ancilliary degrees of freedom to ensure that, for all $\calS_\perp \subseteq [K]$ such that $\calS_\perp \neq \emptyset$, 
    \begin{align}
        \bra{\phi}\bad\ket{\phi} &= 1- \bra{0^m} \left(\prod_{i \in \calS_\perp} V_i\right)^\dagger \ketbra{0^m}{0^m} \prod_{i \in \calS_\perp} V_i \ket{0^m}=1-\left|\bra{0^m} \prod_{i \in \calS_\perp} V_i \ket{0^m}\right|^2 = 1. 
    \end{align}
    This implies that 
    \begin{align}
        \bra{0^m} \prod_{i \in \calS_\perp} V_i \ket{0^m} = 0, \quad \forall \calS_\perp \subseteq [K-1]~\text{ s.t. }\calS_\perp \neq \emptyset,
    \end{align}
    meaning the set of states
    \begin{align}
        \left\{V_1 \ket{0^m},V_2 \ket{0^m}, V_2 V_1 \ket{0^m}, V_3 \ket{0^m},..., \prod_{i\in [K-1]} V_i \ket{0^m}\right\}
    \end{align}
    must all be perpendicular to the state $\ket{0^m}$. 
    
    In order to determine the minimum possible number of measurement ancillae $m$ needed to ensure that all these constraints are satisfied and states are perpendicular, we will reason about the smallest possible Hilbert space in which all these states can be orthogonal. Let $\widetilde{\calH}_i$ be the Hilbert space containing all possible measurement ancilla states orthogonal to $\ket{0^m}$ after measurement $i$. That is,
    \begin{align}
        \widetilde{\calH}_0 &= \spant\{\emptyset\}, 
        \\
        \widetilde{\calH}_1 &= \spant\{V_1 \ket{0^m}\}, \\
        \widetilde{\calH}_2 &= \spant\{V_1 \ket{0^m},V_2 \ket{0^m}, V_2 V_1 \ket{0^m}\}, \\
        \widetilde{\calH}_3 &= \spant\{V_1 \ket{0^m},V_2 \ket{0^m}, V_3 \ket{0^m}, V_2 V_1 \ket{0^m}, V_3 V_2 \ket{0^m}, V_3 V_1 \ket{0^m}, V_3 V_2 V_1 \ket{0^m}\} \\
        & \vdots \\
        \widetilde{\calH}_{i+1} &= \spant\{V_{i+1}\ket{0^m}, V_{i+1} \widetilde{\calH}_{i}, \widetilde{\calH}_{i}\} \label{eqn:Hi+1}
    \end{align}
    We will now prove that,
    \begin{align}
        \dim(\widetilde{\calH}_{i+1}) \geq \dim(\widetilde{\calH}_i)+1, \quad \forall i \in [K-1],
    \end{align}
    via proof by contradiction. 
    
    We begin by assuming that $\dim(\widetilde{\calH}_{i+1}) = \dim(\widetilde{\calH}_i)$. Since $\widetilde{\calH}_{i+1} \supseteq \widetilde{\calH}_{i}$, this implies that $\widetilde{\calH}_{i+1} = \widetilde{\calH}_{i}$. Thus, there exists a bijection between $\widetilde{\calH}_{i+1}$ and $ \widetilde{\calH}_{i}$, governed by multiplication by the $V_{i+1}$ unitary matrix and its inverse $V_{i+1}^\dagger$. From \Cref{eqn:Hi+1}, since $V_i \ket{0^m}$ is an element of $\widetilde{\calH}_{i+1}$, via the inverse mapping,  
    $V_i^\dagger V_i \ket{0^m}=\ket{0^m} \in \widetilde{\calH}_i$. This, however, is a contradiction, since $\widetilde{\calH}_i$ was defined to be the set of all possible measurement ancilla states \emph{orthogonal} to $\ket{0^m}$ after measurement $i$.

    Following from our proof, after $K-1$ measurements the dimension of Hilbert space orthogonal to $\ket{0^m}$ is lower-bounded as $\dim(\widetilde{\calH}_{K-1}) \geq K-1$. Note, however, that the full Hilbert space of the measurement ancillae, $\calH_{K-1}$ must also contain $\ket{0^m}$ state, meaning
    \begin{align}
        \calH_{K-1} = \spant\{\ket{0^m}, \widetilde{\calH}_{K-1}\}, \quad \text{s.t. } \dim(\calH_{K-1}) \geq K.
    \end{align}
    Therefore, the number of measurement ancillae needed to produce an $\calH_{K-1}$ Hilbert space is
    \begin{align}
        m \geq \lceil\log_2 K\rceil,
    \end{align}
    obtaining our desired lower-bound.
\end{proof}

\subsection{Approximate Compression Gadgets}
In the standard block encodings literature, oftentimes creating exact block encodings can be too challenging or demanding. As such, there exists a notion of \emph{approximation} for block encodings. In particular, the matrix $U_B$ is considered an $\epsilon$-precise block encoding of matrix $B$ if
\begin{align}
    \| B - \bra{0^a} U_B \ket{0^a} \| \leq \epsilon,
\end{align}
where $\| \cdot \|$ denotes the operator norm. Inspired by this measure of ``closeness'' to a desired block encoding, in this section, we will introduce the notions of \emph{approximate} multiplication of block encodings and \emph{approximate} compression gadgets. Leveraging these, we will show that in certain regimes of multiplication of block encodings, approximate compression gadgets can be constructed that achieve high-precision multiplication of block encodings while requiring only a constant number of measurement ancillae---surpassing the previously established EMBE lower-bound.

We begin by defining \emph{approximate} multiplication of block encodings (AMBE) and \emph{approximate} compression gadgets (ACGs) as follows.
\begin{definition}[$\epsilon$-Approximate Multiplication of Block Encodings ($\epsilon$-AMBE)]
    We will refer to an $\circclass_{\vec{U}_A, K,m} (\vec{V}, Q)$ circuit as an  \textbf{approximate multiplication of block encodings ($\epsilon$-AMBE)} if
    \begin{align} 
        \| A_{[K]} - \bra{0^{m+a}} \circclass_{\vec{U}_A, K,m} (\vec{V}, Q) \ket{0^{m+a}} \| \leq \epsilon,
    \end{align}
    thereby implementing an $\epsilon$-precise block encoding of the desired  multiplication $A_{[K]}$.
\end{definition}
\noindent Note that from this definition, it follows that for any input $\ket{\psi}$, the output state of an $\epsilon$-approximate multiplication of block encodings circuit is $\epsilon$-close to the desired $A_{[K]}\ket{\psi}$ output state.
\begin{claim} \label{thm:closeness}
    Let $U_{\widetilde{A}_{[K]}}$ denote the unitary implemented by an $\epsilon$-approximate multiplication of block encodings circuit, such that $\widetilde{A}_{[K]} = \bra{0^{m+a}} U_{\widetilde{A}_{[K]}} \ket{0^{m+a}}$. For an arbitrary input state $\ket{\psi}$, let $\ket*{\widetilde{\phi}} = \widetilde{A}_{[K]} \ket{\psi}$ denote the post-selected $\epsilon$-approximate circuit output and let $\ket{\phi} = A_{[K]} \ket{\psi}$ denote the desired exact output. Then, the approximate and desired outputs have fidelity
    \begin{align} \label{eqn:fidelity}
        |\bra{\phi} \ket*{\widetilde{\phi}} |^2 \geq 1-\epsilon^2.
    \end{align}
\end{claim}
\begin{proof}[Proof of \Cref{thm:closeness}]
    From the definition of $\epsilon$-AMBE, we have that
    \begin{align}
        \| A_{K} - \widetilde{A}_{[K]}\| \leq \epsilon.
    \end{align}
    Since the input state to the circuit $\ket{\psi}$ is normalized, i.e. $\|\ket{\psi}\|=1$, this implies that 
    \begin{align}
        \| \ket{\phi}- \ket*{\widetilde{\phi}} \| = \| A_{K} \ket{\psi}- \widetilde{A}_{[K]} \ket{\psi}\| \leq \| A_{K} - \widetilde{A}_{[K]}\|\cdot \|\ket{\psi}\| \leq \epsilon.
    \end{align}
    By the properties of the operator norm,
    \begin{align}
        \epsilon^2 \geq \| \ket{\phi}- \ket*{\widetilde{\phi}} \|^2 = (\bra{\phi}- \bra*{\widetilde{\phi}})(\ket{\phi}- \ket*{\widetilde{\phi}}) = 2 - 2\Re{\bra{\phi}\ket*{\widetilde{\phi}}}, 
    \end{align}
    which implies that 
    \begin{align}
        |\bra{\phi}\ket*{\widetilde{\phi}}| \geq \Re{\bra{\phi}\ket*{\widetilde{\phi}}} \geq 1-\frac{\epsilon^2}{2}.
    \end{align}
    Via some arithmetic manipulation, we obtain the desired lower-bound:
    \begin{align}
        |\bra{\phi}\ket*{\widetilde{\phi}}|^2 \geq \left(1-\frac{\epsilon^2}{2}\right)^2=1-\epsilon^2 + \frac{\epsilon^4}{4} \geq 1-\epsilon^2.
    \end{align}
    This concludes the proof.
\end{proof}

In the last section, we established that \emph{exact} MBE can be implemented optimally with $m = \lceil \log_2 K\rceil$ measurement ancillae. Therefore, we will require that an approximate compression gadget use less than this number of measurement ancillae. 
\begin{definition}[$\epsilon$-Approximate Compression Gadget ($\epsilon$-ACG)] 
    For a sequence of $K$ block encoding unitaries $(U_{A_i})_{i=1}^K$,  any $\circclass_{\vec{U}_A, K,m} (\vec{V}, Q)$~circuit which implementing $\epsilon$-approximate multiplication of block encodings using $m<\lceil \log_2 K \rceil$ measurement ancillae is an \textbf{$\epsilon$-approximate compression gadget ($\epsilon$-ACG)}.
\end{definition}
\noindent  With these general definitions, we will now propose a specific compression gadget--the $p$-Modular Addition Compression Gadget ($p$-MACG)--and prove that it is an $\epsilon$-ACG for a large regime of multiplication of block encodings.

\subsubsection{The \texorpdfstring{$p$}{p}-Modular Addition Compression Gadget}
Our proposal for an $\epsilon$-ACG is heavily inspired by the \cite{low2019Hamiltonian} ECG (as described in \Cref{sec:lw_dcg}). Specifically, we re-parameterize the \addm~addition operations of the \cite{low2019Hamiltonian} ECG as $\addp$ addition operations, where $p\in \mathbb{Z}$ is a constant (relative to $K$). We refer to this compression gadget as the $p$-Modular Addition Compression Gadget ($p$-MACG), since instead of doing addition modulo $2^m=K$ as in the \cite{low2019Hamiltonian} ECG, we perform addition modulo-$2^p$. The $\addp$ operation is implementable with only $p=\calO(1)$ qubits, which implies that the $p$-MACG requires only $m=p<\lceil \log_2 K\rceil$ measurement ancillae. Formally, the $p$-MACG is defined within the \circclass~circuit class as follows.

\begin{definition}[$p$-Modular Addition Compression Gadget ($p$-MACG)]
    For constant $p>1$ and a sequence of $K$ block encoding unitaries $(U_{A_i})_{i=1}^K$, with $\eta_{\max}=\calO(1/K)$, the \textbf{$p$-Modular Addition Compression Gadget ($p$-MACG)} is an $\circclass_{\vec{U}_A, \vec{\Pi}} (\vec{V}, Q)$~circuit composed of: $m=p=\calO(1)$ measurement ancillae, $Q=I_m$, and   
    \begin{align}
        V_i  = \addp, \quad \forall i \in [K-1],
    \end{align}
    where $\addp \ket{x} = \ket{x + 1 \bmod{2^p}}$.
\end{definition}
\noindent All that remains is to show that the $p$-MACG achieves an $\epsilon$-precision block encoding of the desired multiplication and is, thus, an $\epsilon$-ACG.

\subsubsection{Surpassing the Ancilla Lower-Bound for Exact Compression Gadgets}

We will now show that for a broad regime of multiplication of block encoding sequences, the $p$-MACG is an $\epsilon$-ACG where $\epsilon = \calO(1/K^{2^p})$.

We begin by specifying the regime of multiplication of block encoding sequences that we will consider. Suppose we have a sequence of $K$ block encoding unitaries $(U_{A_i})_{i=0}^{K-1}$, such that
\begin{align}
    A_i = \bra{0^a} U_{A_i} \ket{0^a},
\end{align} 
where $A_i$ acts on $n$ qubits and  $U_{A_i}$ acts on $n+a$ qubits. Performing an asymptotic number of multiplications of block encodings only makes sense if each block encoding is fairly close to the identity, i.e., for all $i$,
\begin{align} \label{eqn:id_constrain}
    \|U_{A_i}-I\| \leq \eta_i
\end{align}
for some small $\epsilon_i << 1$. Since $\eta_i$ captures the deviation of the $i^\text{th}$ block encoding matrix from the identity, we will refer to $\eta_i$ as the $i^\text{th}$ ``deviation coefficient''. Block encodings of this form can be decomposed as
\begin{align} \label{eqn:block_enc}
    U_{A_i} = \begin{bmatrix}
        \bra{0^a} U_{A_i} \ket{0^a} & \bra{0^a} U_{A_i} \ket{\perp} \\
        \bra{\perp} U_{A_i} \ket{0^a} & \bra{\perp} U_{A_i} \ket{\perp}
    \end{bmatrix} = \begin{bmatrix}
        A_i & B_i \\
        C_i & D_i
    \end{bmatrix}.
\end{align}
This implies that
\begin{align}
    U_{A_i} = \ketbra{0^a}{0^a} \otimes A_i + \ketbra{0^a}{\perp} \otimes B_i + \ketbra{\perp}{0^a} \otimes C_i +  \ketbra{\perp}{\perp} \otimes D_i 
\end{align}
where \Cref{eqn:id_constrain} imposes that
\begin{align} \label{eqn:op_norm}
    1-\eta_i \leq \| A_i \| \leq 1+\eta_i, \quad \quad \| B_i \| \leq \eta_i, \quad \| C_i \| \leq \eta_i, \quad 1-\eta_i \leq \| D_i \| \leq 1+\eta_i.
\end{align}
Finally, we will denote the maximal deviation coefficient among all $K$ block encodings as
\begin{align}
    \eta_{\text{max}} = \max_{i \in [K]} \eta_i.
\end{align}

With this new notation, we will now prove that in the regime of MBEs where $\eta=\calO(1/K)$, (despite only requiring a constant number of measurement ancillae) a $p$-ACG can implement an $\epsilon$-precise block-encoding of the desired multiplication with $\epsilon=\calO(1/K^{2^p})$.
\begin{theorem} \label{thm:main_mbe}
    Let $(U_{A_i})_{i=1}^K$ be a sequence of $K$ block encodings, with maximal deviation coefficent $\eta_{\max}=\frac{c}{K}$ for some constant $c\in \mathbb{R}$.
    For any $p\in \mathbb{Z}$ such that $p=\calO(1)$, let $U_{pACG}$ be the unitary implemented by the $p$-ACG on $(U_{A_i})_{i=1}^K$. The  $p$-ACG obtains an $\epsilon$-precise block encoding of the desired multiplication sequence $A_{[K]}$, i.e.
    \begin{align}
        \| A_{[K]} - \bra{0^{p+a}} U_{pACG} \ket{0^{p+a}} \| \leq \epsilon, 
    \end{align}
    with inverse-polynomial precision in $K$:
    \begin{align}
        \epsilon = 2 e^c \cdot \left(\frac{e \cdot c^2 }{K\cdot 2^p}\right)^{2^p} = \calO\left(\frac{1}{K^{2^p}}\right).
    \end{align}
\end{theorem}
\begin{proof}[Proof of \Cref{thm:main_mbe}]
    For ease of notation, we will denote the block encoding error as
    \begin{align}
        \err(U_{pACG}) &= \| A_{[K]} - \bra{0^{m+a}} U_{pACG} \ket{0^{m+a}} \|.
    \end{align}
    Thus, the block encoding achieved by the $p$-ACG can be decomposed as:
    \begin{align}
        \bra{0^{p+a}}& U_{pACG} \ket{0^{p+a}} \\
        &= \sum_{x \in \{0,1\}^{K-1}} \bra{0^p}~\addp^{|x|}~\ket{0^p} \cdot \bra{0^a} U_{A_K}\left(\prod_{i=1}^{K-1} \pzeroa^{1-x_i} \pperp^{x_i} U_{A_i}\right) \ket{0^a} \\
        &=  \braket{0^p}{0^p} \cdot \bra{0^a} U_{A_K} \left(\prod_{i=1}^{K-1} \pzeroa U_{A_i}\right) \ket{0^a} \\
        & \quad + \sum_{x \in \{0,1\}^{K-1}\backslash 0^{K-1}} \bra{0^p}~\addp^{|x|}~\ket{0^p} \cdot  \bra{0^a} U_{A_K} \left(\prod_{i=1}^{K-1} \pzeroa^{1-x_i} \pperp^{x_i} U_{A_i} \right) \ket{0^a} \\
        &= A_{[K]} + \sum_{x \in \{0,1\}^{K-1}\backslash 0^{K-1}} \bra{0^p}~\addp^{|x|}~\ket{0^p} \cdot  \bra{0^a} U_{A_K} \left(\prod_{i=1}^{K-1} \pzeroa^{1-x_i} \pperp^{x_i} U_{A_i} \right) \ket{0^a}. \label{eqn:u_hr}
    \end{align}
    From hereon, we will denote the block encoded matrix corresponding to the bad measurement sequence governed by string $x$ as
    \begin{align}
        \seq = \bra{0^a} U_{A_K} \left(\prod_{i=1}^{K-1} \pzeroa^{1-x_i} \pperp^{x_i} U_{A_i} \right) \ket{0^a} .
    \end{align}
    Furthermore, note that
    \begin{align}
        \bra{0^p}~\addp^{|x|}~\ket{0^p} = \delta \{ |x| = 0 \bmod{2^p}\}.
    \end{align}
    Thus, \Cref{eqn:u_hr} further simplifies to
    \begin{align} \label{eqn:uhr_decomp}
        \bra{0^{p+a}} U_{pACG} \ket{0^{p+a}} &= A_{[K]} + \sum_{x \in \{0,1\}^{K-1}\backslash 0^{K-1}} \seq \cdot \delta \{ |x| = 0 \bmod{2^p}\} \\
        &= A_{[K]} + \sum_{\substack{x \in \{0,1\}^{K-1}\backslash 0^{K-1}:\\|x| = 0 \bmod{2^p}}} \seq, 
    \end{align}
    which implies that the block encoding error can be upper-bounded as
    \begin{align}
        \err(U_{pACG}) &= \| A_{[K]} - \bra{0^{p+a}} U_{pACG} \ket{0^{p+a}} \| = \left\| \sum_{\substack{x \in \{0,1\}^{K-1}\backslash 0^{K-1}:\\|x| = 0 \bmod{2^p}}} \seq  \right\|  \leq  \sum_{\substack{x \in \{0,1\}^{K-1}\backslash 0^{K-1}:\\|x| = 0 \bmod{2^p}}} \| \seq \|.  \label{eqn:be_temp}
    \end{align}
    We will now show that the $\|\seq\|$ terms in this expression have a simple upper-bound in terms of $\eta_{\max}$, $|x|$, and $K$.
    \begin{claim} \label{thm:seq_bound}
        $\|\seq\| \leq (\eta_{\max})^{2|x|}\cdot (1+\eta_{\max})^{K}$
    \end{claim}
    \begin{proof}[Proof of \Cref{thm:seq_bound}]
        For $x \in \{0,1\}^{K-1}\backslash 0^{K-1}$, the $\seq$ terms
        encode all possible ``bad'' multiplication sequences. The length-$(K-1)$ bit-string $x$ encodes whether each measurement outcome was ``good''  (0) or ``bad'' (1). Therefore, in the length-$(K+1)$ bit-string sequence $\widetilde{x}=(0,x,0)$, each time a 0 is followed by a 1, this represents a transition between the good and bad ancilla subspaces. This corresponds to the application of one of the on-diagonal block encoded matrices, i.e. $B_i$ or $C_i$ in \Cref{eqn:block_enc}, which both have operator norm $\leq \eta_i$ (\Cref{eqn:op_norm}). Meanwhile, if a 0 is followed by a 0 or a 1 is followed by a 1, the measurement leaves the system in the good or bad subspace, respectively. This corresponds to the application of one of the off-diagonal block encoded matrices, i.e. $A_i$ or $D_i$ in \Cref{eqn:block_enc}, which both have operator norm $\leq 1+\eta_i$ (\Cref{eqn:op_norm}). 
        
        Let $\Delta x$ denote the length-$K$ bit string that tracks these differences in $\widetilde{x}$, i.e.
        \begin{align}
            \Delta x_i = \begin{cases}
                0, &\text{ if } \widetilde{x}_{i+1} = \widetilde{x}_i\\
                1, &\text{ if } \widetilde{x}_{i+1} \neq \widetilde{x}_i 
            \end{cases}.
        \end{align}
        Therefore, leveraging the fact that $|\Delta x| \leq 2|x|$,
        \begin{align}
             \|\seq\| & \leq (\eta_{\max})^{|\Delta x|}\cdot (1+\eta_{\max})^{K-|\Delta x|} \leq (\eta_{\max})^{2|x|}\cdot (1+\eta_{\max})^{K}.
        \end{align}
        This concludes the proof of \Cref{thm:seq_bound}.
    \end{proof}
    \noindent Plugging this into the block encoding error upper-bound of \Cref{eqn:be_temp},
    \begin{align}
        \err(U_{pACG}) &\leq  \sum_{\substack{x \in \{0,1\}^{K-1}\backslash 0^{K-1}:\\|x| = 0 \bmod{2^p}}} (\eta_{\max})^{2|x|}\cdot (1+\eta_{\max})^{K} \\
        & = (1+\eta_{\max})^{K} \cdot \sum_{\substack{r\in[K-1]:\\r = 0 \bmod{2^p}}} \binom{K-1}{r} (\eta_{\max})^{2r}. \label{thm:ub_temp}
    \end{align}
    We will now derive an upper-bound for the summation in this expression.
    \begin{claim} \label{thm:sum_bound}
        $\sum_{\substack{r\in[K-1]:\\r = 0 \bmod{P}}} \binom{K-1}{r} (\eta_{\max})^{2r} \leq 2 \left(\frac{e \cdot \eta_{\max}^2 \cdot K}{P}\right)^P$
    \end{claim}
    \begin{proof}[Proof of \Cref{thm:sum_bound}] 
    We begin by re-expressing the sum as
        \begin{align} \label{eqn:inter}
            \sum_{\substack{r\in[K-1]:\\r = 0 \bmod{P}}} \binom{K-1}{r} (\eta_{\max})^{2r} = \sum_{j=1}^{\lfloor(K-1)/P\rfloor} \binom{K-1}{j P} (\eta_{\max})^{2jP}.
        \end{align}
        Since $\eta_{\max} < 1$, we note that the largest contribution in this sum comes from the $j=1$ term
        \begin{align}
            T_1 = \binom{K-1}{P} (\eta_{\max})^{2P} \leq \left(\frac{e(K-1)}{P}\right)^P(\eta_{\max})^{2P} = e^P \cdot \left(\eta_{\max}^2 \cdot \frac{K-1}{P}\right)^P.
        \end{align}
        Furthermore, note that the ratio of the $(j+1)$st to the 1st term is:
        \begin{align}
            R_{j+1} &= \frac{\binom{K-1}{(j+1) P} (\eta_{\max})^{2(j+1)P}}{\binom{K-1}{P} (\eta_{\max})^{2P}} = (\eta_{\max})^{2Pj} \cdot  \prod_{i=1}^{jP} \left(\frac{K-P-i}{P+i}\right)  \leq \left(\eta_{\max}^2 \cdot \frac{K-1}{P}\right)^{Pj}.
        \end{align}
        From hereon, we will denote
        \begin{align}
            r = \left(\eta_{\max}^2 \cdot \frac{K-1}{P}\right)^P.
        \end{align}
        With this notation, we have that
        \begin{align}
            T_1 \leq r\cdot e^P, \quad R_{j+1} \leq r^j.
        \end{align}
        Furthermore, since we are in the regime where $\eta_{\max} = \calO(1/K)$, note that $r=\calO(1/K^P)<1$. 
        With this, we can decompose \Cref{eqn:inter} by factoring out the dominating $T_1$ term and upper-bounding the resultant expression via a geometric series  as
        \begin{align}
            \sum_{j=1}^{\lfloor(K-1)/P\rfloor} \binom{K-1}{j P} (\eta_{\max})^{2jP} &= T_1 \left[1+\sum_{j=1}^{\lfloor(K-1)/P\rfloor-1}R_{j+1}\right] \leq T_1 \left[1+\sum_{j\geq1}R_{j+1}\right] \\
            &= T_1 \left[1+\sum_{j\geq1}r^j\right] = \frac{T_1}{1-r} \leq e^P \cdot \frac{r}{1-r}
        \end{align}
        Leveraging the fact that $r=\calO(1/K^P) \leq 1/2$, we achieve the desired bound,
        \begin{align}
            \sum_{j=1}^{\lfloor(K-1)/P\rfloor} \binom{K-1}{j P} (\eta_{\max})^{2jP} \leq e^P \cdot \frac{r}{1-r} \leq e^P \cdot \frac{r}{1-\frac{1}{2}} \leq 2e^Pr \leq 2 \left(\frac{e \cdot \eta_{\max}^2 \cdot K}{P}\right)^P.
        \end{align}
        This concludes the proof of \Cref{thm:sum_bound}.
    \end{proof}
    \noindent Plugging this upper-bound into \Cref{thm:ub_temp}, and leveraging the fact that $(1+x)^K \leq e^{Kx}$ for $x \in (0,1)$, the block encoding error upper-bound simplifies to
    \begin{align}
        \err(U_{pACG}) \leq (1+\eta_{\max})^{K} \cdot 2 \left(\frac{e \cdot \eta_{\max}^2 \cdot K}{P}\right)^P \leq e^{K \cdot \eta_{\max}}\cdot 2 \left(\frac{e \cdot \eta_{\max}^2 \cdot K}{P}\right)^P.
    \end{align}
    Plugging in $\eta_{\max} = c/K$, we obtain the desired error bound, 
    \begin{align}
        \err(U_{pACG}) \leq  2e^c \cdot  \left(\frac{e \cdot c^2}{K \cdot P}\right)^P = 2e^c \cdot  \left(\frac{e \cdot c^2}{K \cdot 2^p}\right)^{2^p}= \calO\left(\frac{1}{K^{2^p}}\right).
    \end{align}
    This concludes the proof of \Cref{thm:main_mbe}.
\end{proof}

\begin{corollary} \label{thm:seq_cor}
    For a given error $\epsilon \in (0,1)$, a $p$-MACG can achieve an $\epsilon$-precise for any multiplication of block encoding sequences $(U_{A_i})_{i=1}^K$ satisfying $\eta_{\max}=c/K$ and 
    \begin{align} \label{eqn:seq_lb}
        K \geq \frac{e c^2}{2^p} \cdot \left(\frac{2}{\epsilon}\right)^{1/2^p}.
    \end{align}
\end{corollary}
\begin{proof}[Proof of \Cref{thm:seq_cor}]
    By \Cref{thm:main_mbe}, we have that $\err(U_{pACG}) \leq \epsilon$ if $\eta_{\max}=c/K$ and
    \begin{align}
        2e^c \cdot  \left(\frac{e \cdot c^2}{K \cdot 2^p}\right)^{2^p} \leq \epsilon.
    \end{align}
    Solving this expression in term of $K$ results in the lower-bound of \Cref{eqn:seq_lb}.
\end{proof}

\subsection{Oblivious Amplitude Amplification for Multiplication of Block Encodings} \label{sec:approx_oaa}
As the number of measurements $K$ in a multiplication of block encoding sequence increases, the probability of the all-good-measurement sequence $\{0^a\}^K$ decreases. Thus, we would like a way to boost the probability of measuring this sequence. 

To this end, we will briefly describe the standard quantum algorithms subroutines of amplitude amplification (AA) [\Cref{sec:aa}] and oblivious amplitude amplification (OAA) [\Cref{sec:oaa}]. We will then show how OAA can be tailored to MBE, so as to boost the probability of measuring the all-good-measurement sequence. In particular, we will show that OAA can boost the probability to $\Omega(1)$ for EMBE [\Cref{sec:oaa_embe}] and  $\Omega(1-\eps^2)$ for $\eps$-AMBE [\Cref{sec:oaa_ambe}].

\subsubsection{Amplitude Amplification (AA)} \label{sec:aa}
Amplitude amplification (AA) was originally proposed by \cite{brassard2000quantum} and has become a standard sub-routine in quantum algorithms.
Suppose,
\begin{align}
    \ket{\psi_0} = \alpha_\text{good}\ket{\psi_\text{good}} + \alpha_\text{bad}\ket{\psi_\text{bad}},
\end{align}
where $|\alpha_\text{good}|^2<<1$ is the probability of being in the desired ``good'' state $\ket{\psi_\text{good}}$ and $|\alpha_\text{bad}|^2=1-|\alpha_\text{good}|^2$ is the probability of being in any orthogonal ``bad'' state $\ket{\psi_\text{bad}}$. AA enables us to map $\ket{\psi_0}$ to a state with large overlap with $\ket{\psi_\text{good}}$ or, in other words, AA amplifies the amplitude of the desired $\ket{\psi_\text{good}}$ state.

In terms of implementation, AA requires access to two reflection oracles. The first is a reflection about the $\ket{\psi_\text{good}}$ state,
\begin{align}
    R_\text{good}=I_n-2\ketbra{\psi_\text{good}}{\psi_\text{good}},
\end{align}
and the second is a reflection with respect to the initial $\ket{\psi_0}$ state,
\begin{align}
    R_{\psi_0}=2\ketbra{\psi_0}{\psi_0}-I_n.
\end{align}
Following a similar procedure to Grover's algorithm, these are used to implement the Grover iterate
\begin{align}
    G = R_{\psi_0}R_\text{good}.
\end{align}
Letting $k=\bigO{1/|\alpha_\text{good}|}$, $G^k\ket{\psi_0}$ maps $\ket{\psi_0}$ to a state with $\Omega(1)$ overlap with $\ket{\psi_\text{good}}$.

\subsubsection{Oblivious Amplitude Amplification (OAA)} \label{sec:oaa}

The main downside of AA is that knowledge of the ``good'' state $\ket{\psi_\text{good}}$ is required in order to implement the reflection oracle $R_\text{good}$. However, if $\ket{\psi_\text{good}}$ is unknown or inefficient to implement, this  requirement can be avoided with block encodings and so-called ``signal'' qubits. The form of AA we will now describe, which is ``oblivious'' to the precise form of the good state, is known as oblivious amplitude amplification (OAA) and was originally introduced by \cite{low2016HamSimQubitization}.

In addition to $n$ qubits encoding the good and bad states, the initial state will also include $m$ signal ancilla qubits, distinguishing the good and bad subspaces. Specifically,
\begin{align} \label{eqn:aa}
    \ket{\psi^m_0} = U_{\psi^m_0} \ket{0^m}\ket{0^n} = \alpha_\text{good}\ket{0^m}\ket{\psi_\text{good}} + \alpha_\text{bad}\ket{\perp}\ket{\psi_\text{bad}}.
\end{align}
The signal qubits have value $\ket{0^m}$ only if the main register is in the good state, $\ket{\psi_\text{good}}$. When the main register is in a bad state, the signal qubits are in a state orthogonal to  $\ket{0^m}$, denoted $\ket{\perp}$. Thus, for $\Pi=\ketbra{0^m}{0^m}$,
\begin{align}
    (\Pi \otimes I_n)\ket{0^m}\ket{\psi_\text{good}} &= \ket{0^m}\ket{\psi_\text{good}} \quad \text{and}\quad 
    (\Pi \otimes I_n)\ket{\perp}\ket{\psi_\text{bad}} = 0,
\end{align}
meaning that the good subspace can be easily identified by measuring $\ket{\psi^m_0}$'s signal ancillae and post-selecting the $0^m$ outcome. In expectation, $1/|\alpha_\text{good}|^2$ measurements of $\ket{\psi^m_0}$ would be required to obtain $\ket{\psi_\text{good}}$. Once again, we can improve this quadratically by leveraging AA.

Since the good and bad states have signal ancillae, the AA reflection oracles can be simplified. Namely, reflection about the good state reduces to reflection only about the signal ancillae,
\begin{align}
    R_\text{good} = (I_m-2\ketbra{0^m}{0^m}) \otimes I_n.
\end{align}
Meanwhile, the reflection about the initial state $\ket{\psi^m_0}$ is about both the main register and ancillae,
\begin{align}
    R_{\psi_0^m} = 2\ketbra{\psi^{m}_0}{\psi^{m}_0}-I_{m+n}.
\end{align}
Thus,  the Grover iterate is defined as $G = R_{\psi_0^m} R_\text{good}$. For $k=\bigO{1/|\alpha_\text{good}|}$, $G^k \ket{\psi^m_0}$ produces a state with $\Omega(1)$ overlap with $\ket{\psi_\text{good}}$.

\subsubsection{OAA for Exact Multiplication of Block Encodings}\label{sec:oaa_embe}

In an MBE circuit, each unitary $U_{A_i}$ has probability $p_i$ of the good $0^a$ ancilla measurement outcome and, thus, probability $1-p_i$ of a bad $\perp$ outcome. Thus, the desired all-good-measurements sequence $\{0^a\}^K$ has probability 
\begin{align} \label{eqn:probs}
    p = \prod_{i=1}^K p_i
\end{align}
which can be quite small depending on the number of measurements $K$ and individual measurement probabilities $p_i$. Furthermore, $\bigO{1/p}$ repetitions of the MBE circuit would be required in expectation. However, via OAA, the success probability can be coherently amplified to a constant value with quadractic queries to an EMBE circuit.

In this case, the OAA signal ancillae are the $m$ measurement ancillae and $a$ block encoding ancillae of the EMBE circuit, which together indicate whether the multiplication sequence was successful or not. Specifically, in the successful case, the main $n$-qubit register is in the desired state $\ket{\psi_\text{good}}=A_{[K]}\ket{\psi}$ and the $(m+a)$ signal ancillae in the state $\ket{0^{m+a}}$. In the unsuccessful case, the $(m+a)$ signal ancillae are in a state $\ket{\perp}$, orthogonal to $\ket{0^a}$, and the main $n$-qubit register is in a junk state $\ket{\psi_\text{bad}}$. 

Formally, denoting the unitary implemented by the EMBE circuit as $U_\embe$,
\begin{align}
    \ket{\psi_\embe} = U_\embe \ket{0^{m+a}} \ket{\psi_0} = \alpha_\text{good} \ket{0^{m+a}}\ket{\psi_\text{good}} + \alpha_\text{bad} \ket{\perp} \ket{\psi_\text{bad}},
\end{align}
where $|\alpha_\text{good}|^2=p$ as defined in \Cref{eqn:probs}.
This state is analogous to the initial state $\ket{\psi_0^m}$ in OAA (\Cref{eqn:aa}). Thus, to run OAA we construct the two reflection oracles: 
\begin{align}
    R_\text{good} &= (I_{m+a}-2\ketbra{0^{m+a}}{0^{m+a}}) \otimes I_{n} \\
    R_{\psi_\embe} &= 2\ketbra{\psi_\embe}{\psi_\embe}-I_{m+a+n}=U_{\psi_\embe}(2\ketbra{0^{m+a},\psi_0}{0^{m+a},\psi_0}-I_{m+a+n})U_{\psi_\embe}^\dagger,
\end{align}
and the Grover iterate $G = R_{\psi_\embe} R_\text{good}$. Letting $k=\calO(1/\sqrt{p})$, then $G^k \ket{\psi_\embe}$ produces a state with $\Omega(1)$ overlap with $\ket{0^{m+a}}\ket{\psi_\text{good}}$, as desired.

\subsubsection{OAA for Approximate Multiplication of Block Encodings} \label{sec:oaa_ambe}
To conclude, we will demonstrate how OAA can be used to coherently boost the probability of the all-good-measurement sequence, even in the setting of \emph{approximate} MBE. The OAA implementation for AMBE does not differ too much from that of EMBE, but the analysis is a bit more involved. In particular, we prove the following lemma.
\begin{lemma}[OAA for $\eps$-AMBE] \label{thm:oaa_ambe}
    Via repeated queries to an $\eps$-approximate multiplication of block encodings circuit, oblivious amplitude amplification can be used to boost the fidelity between the circuit output and the desired state to $1-\calO(\eps^2)$.
\end{lemma}
\begin{proof}[Proof of \Cref{thm:oaa_ambe}]
    Let $U_\ambe$ denote the unitary implemented by an $\epsilon$-AMBE circuit. The key difference from EMBE is that the $\epsilon$-AMBE outputs a state of the form:
    \begin{align} \label{eqn:ambe_aa_state}
        \ket{\psi_\ambe} = U_\ambe \ket{0^{m+a}} \ket{\psi_0} = \alpha_\text{good} \ket{0^{m+a}}\ket{\psi_\text{good}} + \alpha_\text{wrong} \ket{0^{m+a}}\ket{\psi_\text{wrong}} + \alpha_\text{bad} \ket{\perp} \ket{\psi_\text{bad}}.
    \end{align}
    As for EMBE, $\ket{\psi_\text{good}}=A_{[K]}\ket{\psi}$ is the desired good state, with signal ancillae $\ket{0^{m+a}}$, and $\ket{\psi_\text{bad}}$ is a junk state with signal ancillae $\ket{\perp}$ (orthogonal to $\ket{0^{m+a}}$). However, for  $\epsilon$-AMBE there is an additional junk state $\ket{\psi_\text{wrong}}$, which has the same signal ancillae $\ket{0^{m+a}}$ as the good state $\ket{\psi_\text{good}}$. For ease of notation, we will denote the $n$-qubit sub-state with signal ancilla $\ket{0^{m+a}}$ as
    \begin{align}
        \ket{\psi_\text{sig}} = \frac{1}{\alpha_\text{sig}} \cdot \left(\alpha_\text{good} \ket{\psi_\text{good}} + \alpha_\text{wrong} \ket{\psi_\text{wrong}}\right),
    \end{align}
    where $\alpha_\text{sig}=\sqrt{|\alpha_\text{good}|^2+|\alpha_\text{wrong}|^2}$ is a normalization factor. Note that $\ket{\psi_\text{sig}}$ is the state obtained by post-selecting $\ket{\psi_\ambe}$ for the ancillae outcome $0^{m+a}$, i.e.
    \begin{align}
        \ket{\psi_\text{sig}} = (\bra{0^{m+a}} \otimes I_n) U_\ambe \ket{0^{m+a}} \ket{\psi_0} = (\bra{0^{m+a}} \otimes I_n)\ket{\psi_\text{sig}}.
    \end{align}
    Thus, \Cref{eqn:ambe_aa_state} simplifies to 
    \begin{align} 
        \ket{\psi_\ambe} = \alpha_\text{sig} \ket{0^{m+a}}\ket{\psi_\text{sig}} + \alpha_\text{bad} \ket{\perp} \ket{\psi_\text{bad}}.
    \end{align}
    
    In this form, it is clear that OAA can be used to boost the amplitude of the $\ket{\psi_\text{sig}}$ state. Specifically, OAA will be run with the Grover iterate $G =  R_{\psi_\ambe} R_\text{sig}$, where:
    \begin{align}
         R_\text{sig} &= (I_{m+a}-2\ketbra{0^{m+a}}{0^{m+a}}) \otimes I_{n} \\
        R_{\psi_\ambe} &= 2\ketbra{\psi_\ambe}{\psi_\ambe}-I_{m+a+n}= U_{\psi_\ambe}(2\ketbra{0^{m+a},\psi_0}{0^{m+a},\psi_0}-I_{m+a+n})U_{\psi_\ambe}^\dagger.
    \end{align}
    Letting $k=\calO(1/\sqrt{|\alpha_\text{sig}|^2})$, $G^k \ket{\psi_\ambe}$ produces a state with $\Omega(1)$ overlap with $\ket{0^{m+a}}\ket{\psi_\text{sig}}$. Furthermore, by \Cref{thm:closeness}, $\ket{\psi_\text{sig}}$ has fidelity at least $1-\epsilon^2$ with $\ket{\psi_\text{good}}$. Thus, the output state of this OAA procedure has $1-\calO(\eps^2)$ overlap with the desired exact state $\ket{\psi_\text{good}}$.
\end{proof}

\end{document}